\documentclass[11pt]{article}
\pdfoutput=1
\usepackage[T1]{fontenc}
\usepackage{lmodern}
\usepackage{amsthm}
\usepackage{graphicx} 
\usepackage{array} 

\usepackage{amsmath, amssymb, amsfonts, verbatim}
\usepackage{hyphenat,subfigure,multirow}
\usepackage{algorithm,algpseudocode}
\usepackage{algorithmicx}
\usepackage{algpseudocode}
\usepackage{mathtools}
\usepackage{stmaryrd}

\usepackage[usenames,dvipsnames]{xcolor}

\usepackage{tcolorbox}
\tcbuselibrary{skins,breakable}
\tcbset{enhanced jigsaw}

\definecolor{DarkRed}{rgb}{0.5,0.1,0.1}
\definecolor{DarkBlue}{rgb}{0.1,0.1,0.5}

\usepackage{hyperref}
\hypersetup{
colorlinks=true,
pdfnewwindow=true,
citecolor=ForestGreen,
linkcolor=DarkRed,
filecolor=DarkRed,
urlcolor=DarkBlue
}

\usepackage{bm}
\usepackage{url}
\usepackage{xspace}
\usepackage[mathscr]{euscript}

\usepackage{mdframed}

\usepackage{cite}
\usepackage{enumitem}

\usepackage[margin=1in]{geometry}

\newtheorem{theorem}{Theorem}
\newtheorem{lemma}{Lemma}

\newtheorem{cor}{Corollary}

\theoremstyle{definition}
\newtheorem{rem}{Remark}
\newtheorem{definition}{Definition}

\newtheorem*{claim*}{Claim}
\newtheorem*{proposition*}{Proposition}
\newtheorem*{lemma*}{Lemma}
\newtheorem*{problem*}{Problem}

\newtheorem{mdresult}[theorem]{Theorem}

\newtheorem{mdinvariant}{Invariant}

\newcommand{\istrut}[2][0]{\rule[- #1 mm]{0mm}{#1 mm}\rule{0mm}{#2 mm}}

\newcommand{\ignore}[1]{}

\newcommand{\Var}{\operatorname{Var}}

\newcommand{\E}{\mathbb{E}}
\newcommand{\ceil}[1]{\left\lceil{#1}\right\rceil}
\newcommand{\floor}[1]{\left\lfloor{#1} \right\rfloor}
\newcommand{\bydef}{\stackrel{\mathrm{def}}{=}}

\newcommand{\paren}[1]{\mathopen{}\left(#1\right)\mathclose{}}
\newcommand{\fish}{\mathsf{Fish}}
\newcommand{\fishmonger}{\textsf{Fishmonger}}

\newcommand{\MVP}{\textsf{MVP}}
\newcommand{\Martingale}{\textsf{Martingale}}
\newcommand{\Curtain}{\textsf{Curtain}}
\newcommand{\sCurtain}{\textsf{SecondCurtain}}

\newcommand{\Indicator}[1]{\left\llbracket\istrut[1]{3.5} #1 \right\rrbracket}

\newcommand{\LL}{\mathsf{LL}}

\newcommand{\HyperLogLog}{\textsf{HyperLogLog}}
\newcommand{\LogLog}{\textsf{LogLog}}
\newcommand{\qLL}{q\text{-}\LL}
\newcommand{\PCSA}{\textsf{PCSA}}
\newcommand{\qPCSA}{q\text{-}\PCSA}
\newcommand{\MinCount}{\textsf{MinCount}}

\newcommand{\hS}{S^{\operatorname{h}}}
\newcommand{\lS}{S^{\operatorname{l}}}

\newcommand{\Cell}{\operatorname{Cell}}

\allowdisplaybreaks

\title{Non-Mergeable Sketching 
for Cardinality Estimation\thanks{This work was supported by NSF grants CCF-1637546 and CCF-1815316.}}

\author{Seth Pettie\\ University of Michigan \\ \footnotesize
\texttt{pettie@umich.edu} \and
Dingyu Wang\\ University of Michigan\\
\footnotesize \texttt{wangdy@umich.edu} \and
Longhui Yin\\ Tsinghua University\\
\footnotesize \texttt{ylh17@mails.tsinghua.edu.cn}
}

\date{}

\begin{document}
\maketitle

\begin{abstract}
\emph{Cardinality estimation} is perhaps the simplest non-trivial statistical problem that can be solved via sketching.
Industrially-deployed sketches like \textsf{HyperLogLog}, \textsf{MinHash}, and \textsf{PCSA} are \emph{mergeable}, which means that large data sets can be sketched in a distributed environment,
and then merged into a single sketch of the whole data set.
In the last decade a variety of sketches have been developed 
that are \emph{non-mergeable}, but attractive for other reasons.
They are \emph{simpler}, their cardinality estimates are \emph{strictly unbiased}, and they have substantially \emph{lower variance}.

\bigskip

We evaluate sketching schemes on a reasonably level playing field, in terms of their \emph{memory-variance product} (\MVP). 
I.e., a sketch that occupies $5m$ bits and whose relative variance is $2/m$ (standard error $\sqrt{2/m}$) has an \MVP{} of $10$.  Our contributions are as follows.

\begin{itemize}
    \item Cohen~\cite{Cohen15} and Ting~\cite{Ting14} independently discovered what we call the \emph{\Martingale{} transform} for converting a mergeable 
    sketch into a non-mergeable sketch.
    We present a simpler way to analyze the limiting \MVP{} of \Martingale-type sketches.
    
    \item Pettie and Wang proved that the \fishmonger{} sketch~\cite{PettieW21} has the best $\MVP$, $H_0/I_0 \approx 1.98$, among a class of mergeable sketches called ``linearizable'' sketches.
    ($H_0$ and $I_0$ are precisely defined constants.)
    We prove that the \Martingale{} transform is optimal in the non-mergeable world, and that \Martingale{} \fishmonger{} in particular is optimal among linearizable sketches, with an 
    $\MVP$ of $H_0/2 \approx 1.63$.  
    E.g., this is circumstantial evidence that to achieve 1\% standard error, we cannot do better than a 2 kilobyte sketch.
    
    \item \Martingale{} \fishmonger{} is neither simple nor practical.  We develop a new mergeable sketch called \Curtain{} that strikes a nice balance between simplicity and efficiency, and prove that \Martingale{} \Curtain{} has limiting $\MVP\approx 2.31$.  It can be updated with $O(1)$ memory accesses and it has lower empirical variance 
    than \Martingale{} \LogLog, 
    a practical non-mergeable version of \HyperLogLog.
\end{itemize}
\end{abstract}

\thispagestyle{empty}
\setcounter{page}{0}

\newpage

\section{Introduction}\label{sect:introduction}

\emph{Cardinality estimation}\footnote{(aka $F_0$ estimation or \emph{Distinct Elements})} 
is a fundamental problem in streaming and sketching
with diverse applications in databases~\cite{ChiaDPSLDWG19,Freitag19}, 
network monitoring~\cite{Ben-BasatEFMR18,BuddhikaMPP17,XiaoCZCLLL17,ChenCC17}, 
nearest neighbor search~\cite{Pham17},
caching~\cite{WiresIDHW14}, and genomics~\cite{MarcaisSPK19,ElworthWKBCBGBST20,WoodLL19,BakerL19}.
In the \emph{\textbf{sequential}} setting of this problem, 
we receive the elements of a multiset 
$\mathscr{A} = \{a_1,a_2,\ldots,a_N\}$ one at a time.
We maintain a small \emph{sketch} $S$ of the elements seen so far,
such that the true cardinality $\lambda = |\mathscr{A}|$
is estimated by
some
$\hat{\lambda}(S)$. The \emph{\textbf{distributed}} setting is similar, except 
that $\mathscr{A}$ is partitioned arbitrarily among several
machines, the shares being sketched separately and combined 
into a sketch of $\mathscr{A}$.  Only \emph{\textbf{mergeable}}
sketches are deployed in distributed settings; see Definition~\ref{def:mergeable} below.

\begin{definition}\label{def:models}
In the \textsc{random oracle model} $\mathscr{A}\subseteq [U]$
and we have oracle access to a uniformly random permutation
$h:[U]\rightarrow [U]$ (or a uniformly random hash function $h:[U]\rightarrow [0,1]$).
In the \textsc{standard model} we can generate random bits as necessary,
but must explicitly store any hash functions in the sketch.
\end{definition}

\begin{definition}\label{def:mergeable}
Suppose $\mathscr{A}^{(1)},\mathscr{A}^{(2)}$ are multisets such that
$\mathscr{A} = \mathscr{A}^{(1)}\cup \mathscr{A}^{(2)}$.
A sketching scheme is \emph{\textbf{mergeable}}
if, whenever, $\mathscr{A}^{(1)},\mathscr{A}^{(2)}$ are sketched as
$S^{(1)},S^{(2)}$ (using the same \textsc{random oracle} $h$ or the 
same source of random bits in the \textsc{standard model}), 
the sketch $S$ of $\mathcal{A}$ can be computed from $S^{(1)},S^{(2)}$ alone.
\end{definition}

\textsc{Standard model} 
sketches~\cite{AlonGMS99,Bar-YossefJKST02,Bar-YossefKS02,Blasiok20,GibbonsT01,KaneNW10}
usually make an $(\epsilon,\delta)$-guarantee, i.e.,
\[
\Pr\left(\hat{\lambda} \not\in [(1-\epsilon)\lambda,(1+\epsilon)\lambda]\right) < \delta.
\]
The state-of-the-art \textsc{standard model} sketch~\cite{Blasiok20,KaneNW10} uses
$O(\epsilon^{-2}\log\delta^{-1} + \log U)$ bits, which is optimal 
at this level of specificity, as it meets the space lower bounds of 
$\Omega(\log U)$, $\Omega(\epsilon^{-2})$ (when $\delta=\Theta(1)$), 
and $\Omega(\epsilon^{-2}\log\delta^{-1})$~\cite{AlonGMS99,IndykW03,JayramW13}.
However, the leading constants hidden by~\cite{Blasiok20,KaneNW10} are quite large.

In the \textsc{random oracle model} 
the cardinality estimate $\hat{\lambda}$ 
typically has negligible bias, and errors
are expressed in terms of the \emph{relative 
variance} $\lambda^{-2}\cdot \Var(\hat{\lambda}\mid\lambda)$
or \emph{relative standard deviation} $\lambda^{-1} \sqrt{\Var(\hat{\lambda}\mid \lambda)}$,
also called the \emph{standard error}.  Sketches that use $\Omega(m)$ bits typically have
relative variances of $O(m)$.  
Thus, the most natural way to measure the quality of the 
\emph{sketching scheme} itself is to look at its limiting 
\emph{memory-variance product} ($\MVP$),
i.e., the product of its memory and variance as $m\rightarrow \infty$.

Until about a decade ago, all \textsc{standard/random oracle} sketches were mergeable,
and suitable to both distributed and sequential applications.
For reasons that are not 
clear to us, 
the idea of \emph{\textbf{non-mergeable}} sketching
was discovered independently by multiple groups~\cite{ChenCSN11,HelmiLMV12,Cohen15,Ting14} at about the same time, and quite \emph{late} in the 40-year history of cardinality estimation.
Chen, Cao, Shepp, and Nguyen~\cite{ChenCSN11} invented the \textsf{S-Bitmap} in 2011, 
followed by Helmi, Lumbroso, Mart{\'i}nez, and Viola's\cite{HelmiLMV12} 
\textsf{Recordinality} in 2012.
In 2014 Cohen~\cite{Cohen15} and Ting~\cite{Ting14} independently 
invented what we call the \emph{\Martingale{} transform}, which is
a simple, mechanical way to transform any mergeable sketch into a (better)
non-mergeable sketch.\footnote{Cohen~\cite{Cohen15} called
these \emph{Historical Inverse Probability} (HIP) sketches 
and Ting~\cite{Ting14} applied the prefix \emph{Streaming} 
to emphasize that they can be used in the single-stream setting,
not the distributed setting.}

In a companion paper~\cite{PettieW21}, we analyzed the $\MVP$s 
of \emph{mergeable} sketches under the assumption that the
sketch was compressed to its entropy bound.  \fishmonger{}
(an entropy compressed variant of \PCSA{} with a different estimator function)
was shown to have $\MVP = H_0/I_0 \approx 1.98$, where 
\begin{align*}
    H_0 &= (\ln 2)^{-1} + \sum_{k=1}^\infty k^{-1}\log_2(1+1/k)
&\mbox{ and \ } I_0 &= \zeta(2) = \pi^2/6.
\end{align*}
Furthermore, $H_0/I_0$ was shown to be the minimum $\MVP$ 
among \emph{linearizable} sketches, a subset of \emph{mergeable}
sketches that includes all the popular sketches (\HyperLogLog, \PCSA,
\textsf{MinHash}, etc.).  

Our aim in \emph{this} paper is to build a useful framework for designing 
and analyzing \emph{non-mergeable} sketching schemes, and, following~\cite{PettieW21},
to develop a theory of space-variance optimality in the non-mergeable world.
We work in the \textsc{random oracle model}.
Our results are as follows.
\begin{itemize}
    \item Although the \emph{\Martingale{} transform} itself is simple,
    analyzing the variance of these sketches is not.  For example,
    Cohen~\cite{Cohen15} and Ting~\cite{Ting14} estimated the standard error
    of \Martingale{} \LogLog{} to be about $\approx \sqrt{3/(4m)} \approx 0.866/\sqrt{m}$ 
    and about $\approx 1/(2\alpha_m m)$, respectively, where 
    the latter tends to $\sqrt{\ln 2/m} \approx 0.8326/\sqrt{m}$ as $m\rightarrow \infty$.\footnote{Here $\alpha_m = \left(m\int_0^\infty \left(\log_2\left(\frac{2+u}{1+u}\right)\right)^m du\right)^{-1}$ 
    is the coefficient of Flajolet et al.'s \textsf{HyperLogLog} estimator.}
    We give a general method for determining the limiting relative variance of \Martingale{}
    sketches that is strongly influenced by Ting's perspective.
    \item What is the most efficient (smallest $\MVP$) non-mergeable sketch for
    cardinality estimation?  The best \Martingale{} sketches perform better than the 
    \emph{ad hoc} non-mergeable \textsf{S-Bitmap} and \textsf{Recordinality}, but
    perhaps there is a completely different, better way to systematically 
    build non-mergeable sketches.
    We prove that up to some natural assumptions\footnote{(the sketch is insensitive to duplicates,
    and the estimator is unbiased)} the best non-mergeable sketch is a
    \Martingale{} \textsf{X} sketch, for some \textsf{X}.  Furthermore, 
    we prove that \Martingale{} \fishmonger, 
    having $\MVP$ of $H_0/2 \approx 1.63$, 
    is optimal 
    among all \Martingale{} \textsf{X} sketches, where \textsf{X} is \emph{linearizable}.
    This provides some circumstantial evidence that \Martingale{} \fishmonger{} is optimal,
    and that if we want, say, 1\% standard error, we need to use a 
    $H_0/2\cdot (0.01)^{-2}$-bit sketch,  $\approx $ 2 kilobytes.
    \item \Martingale{} \fishmonger{} has an attractive $\MVP$, but it is slow and cumbersome
    to implement.  We propose a new mergeable sketch called \Curtain{} that is
    ``naturally'' space efficient and easy to update in $O(1)$ memory accesses,
    and prove that \Martingale{} \Curtain{} has a limiting $\MVP \approx 2.31$.
\end{itemize}

\subsection{Prior Work: Mergeable Sketches}

Let $S_i$ be the state of the sketch after processing
$(a_1,\ldots,a_i)$.

The state of the \PCSA{} sketch~\cite{FlajoletM85} is a 2D matrix
$S\in \{0,1\}^{m\times \log U}$ and the
hash function 
$h : [U]\rightarrow [m]\times \mathbb{Z}^+$ 
produces two indices: $h(a) = (j,k)$
with probability $m^{-1}2^{-k}$.
$S_i(j,k)=1$ iff $\exists i'\in [i]. h(a_{i'})=(j,k)$.
Flajolet and Martin~\cite{FlajoletM85}
proved that a certain estimator has
standard error $0.78/\sqrt{m}$, making 
the $\MVP$ around $(0.78)^2\log U \approx 0.6\log U$.

Durand and Flajolet's \LogLog{} sketch~\cite{DurandF03}
consists of $m$ counters.  
It interprets $h$ exactly as in \PCSA,
and sets $S_i(j)=k$ iff $k$ is maximum such 
that $\exists i'\in [i]. h(a_{i'}) = (j,k)$.
Durand and Flajolet's estimator is 
of the form 
$\hat{\lambda}(S) \propto m2^{m^{-1}\sum_j S(j)}$
and has standard error $\approx 1.3/\sqrt{m}$.
Flajolet, Fusy, Gandouet, and Meunier's \HyperLogLog~\cite{FlajoletFGM07} 
is the same sketch but with the estimator
$\hat{\lambda}(S) \propto 
m^2(\sum_{j} 2^{-S(j)})^{-1}$.
They proved that it has standard error
tending to $\approx 1.04/\sqrt{m}$.
As the space is $m\log\log U$ bits, 
the $\MVP$ is $\approx 1.08\log\log U$.

The \MinCount{} sketch (aka \textsf{MinHash} or \textsf{Bottom}-$m$~\cite{Cohen97,CohenK08,Broder97}) 
stores the smallest
$m$ hash values, which we assume requires
$\log U$ bits each.  Using an appropriate estimator~\cite{Giroire09,ChassaingG06,Lumbroso10},
the standard error is $1/\sqrt{m}$ and $\MVP{} = \log U$.

It is straightforward to see that the entropy of \PCSA{} 
and \LogLog{} are both $\Theta(m)$.  Scheuermann and Mauve~\cite{ScheuermannM07} experimented with entropy
compressed versions of \PCSA{} and \HyperLogLog{} and 
found \PCSA{} to be slightly superior.  Rather than 
use the given estimators of~\cite{FlajoletM85,FlajoletFGM07,DurandF03},
Lang~\cite{Lang17} used Maximum Likelihood-type
Estimators and found entropy-compressed \PCSA{} 
to be significantly better than entropy-compressed 
\LogLog{} (with MLE estimators).  
Pettie and Wang~\cite{PettieW21}
defined the \underline{Fi}sher-\underline{Sh}annon ($\fish$)\footnote{$\fish$ is essentially the same as $\MVP$,
under the assumption that the sketch state is compressed to its entropy.}
number of a sketch as the ratio of its Shannon 
entropy (controlling its entropy-compressed size) to 
its Fisher information 
(controlling the variance of a statistically efficient estimator), 
and proved that the $\fish$-number of any base-$q$ 
\PCSA{} is $H_0/I_0$, and that the $\fish$-number of 
base-$q$ \LogLog{} is worse, 
but tends to $H_0/I_0$ in the limit as $q\rightarrow \infty$.
Here $H_0$ and $I_0$ are:
\begin{align*}
    H_0 &= (\ln 2)^{-1} + \sum_{k=1}^\infty k^{-1}\log_2(1+1/k)\\
\mbox{ and \ } I_0 &= \zeta(2) = \pi^2/6.
\end{align*}

\begin{table}[]
    \centering
    \begin{tabular}{|l|l|l|}
\multicolumn{1}{l}{\textsc{Mergeable Sketch}} &
\multicolumn{1}{l}{\textsc{Limiting} $\MVP$} &
\multicolumn{1}{l}{\textsc{Notes}}\\\hline
\textsf{PCSA}       \hfill \cite{FlajoletM85} & $.6\log U \approx 38.9$          & For $U=2^{64}$\\
\textsf{LogLog}     \hfill \cite{DurandF03} & $1.69\log\log U \approx 10.11$   & For $U=2^{64}$\\
\textsf{MinCount} \hfill \cite{Giroire09,ChassaingG06,Lumbroso10} & $\log U = 64$    & For $U=2^{64}$\\
\textsf{HyperLogLog} \hfill \cite{FlajoletFGM07} & $1.08\log\log U \approx 6.48$    & For $U=2^{64}$\\
\textsf{Fishmonger} \hfill \cite{PettieW21} & $H_0/I_0 \approx 1.98$            & \\\hline\hline
\multicolumn{2}{l}{\textsc{Non-Mergeable Sketch}\istrut[0]{4.5}}\\\hline
\textsf{S-Bitmap} \hfill \cite{ChenCSN11} &  $O(\log^2(U/m))$                           & \\
\textsf{Recordinality} \hfill \cite{HelmiLMV12} & $O(\log(\lambda/m)\log U)$                      & \\
\Martingale{} $\PCSA$ \hfill {\bf new} & $0.35\log U \approx 22.4$& For $U=2^{64}$ \\
\Martingale{} \textsf{LogLog} \hfill \cite{Cohen15,Ting14} & $0.69\log\log U \approx 4.16$ &  For $U=2^{64}$\\
\Martingale{} \textsf{MinCount} \hfill \cite{Cohen15,Ting14} & $0.5\log U = 32$ &  For $U=2^{64}$ \\
\Martingale{} \fishmonger{} \hfill {\bf new} & $H_0/2\approx 1.63$ & $H_0 = (\ln 2)^{-1}+\sum_{k\geq 1} \frac{\log_2(1+1/k)}{k}$  \\\hline
\Martingale{} \Curtain{} \hfill {\bf new}& $\approx 2.31$ & Theorem~\ref{thm:martingale-curtain} with $(q,a,h)=(2.91,2,1)$\\\hline\hline
\multicolumn{2}{l}{\textsc{Non-Mergeable Lower Bound}\istrut[0]{4.5}} \\\hline
\Martingale{} \textsf{X} \hfill {\bf new} & $\geq H_0/2$    & \textsf{X} is a \emph{linearizable} sketch\\\hline\hline
    \end{tabular}
    \caption{A selection of results on composable sketches (top)
    and non-composable \Martingale{} sketches (bottom) in terms of their
    limiting memory-variance product ($\MVP$).  Logarithms are base 2.}
    \label{tab:results}
\end{table}

\subsection{Prior Work: Non-Mergeable Sketches}

Chen, Cao, Shepp, and Nguyen's \textsf{S-Bitmap}~\cite{ChenCSN11} 
consists of a bit string $S\in \{0,1\}^m$
and $m$ known constants 
$0 \leq \tau_0 < \tau_1 < \cdots < \tau_{m-1} < 1$.
It interprets $h(a) = (j,\rho) \in [m]\times [0,1]$ as an index $j$ 
and real $\rho$ and when processing $a$, sets $S(j)\leftarrow 1$
iff $\rho > \tau_{\operatorname{HammingWeight}(S)}$.
One may confirm that $S$ is insensitive to duplicates in the
stream $\mathcal{A}$, but its state depends on the \emph{order}
in which $\mathcal{A}$ is scanned.  By setting the $\tau$-thresholds
and estimator properly, the standard error is 
$\approx \ln(eU/m)/(2\sqrt{m})$ and $\MVP = O(\log^2(U/m))$.

\textsf{Recordinality}~\cite{HelmiLMV12} is based on \MinCount;
it stores $(S,\operatorname{cnt})$, where $S$ is the $m$ smallest hash values encountered and $\operatorname{cnt}$ is the \emph{number of times} that $S$ has changed.  The estimator looks only at $\operatorname{cnt}$, not $S$, and has standard error
$\approx \sqrt{\ln(\lambda/em)/m}$
and $\MVP{} = O(\log(\lambda/m)\log U)$.

Cohen~\cite{Cohen15} and Ting~\cite{Ting14} independently 
described how to turn any sketch into a non-mergeable sketch 
using what we call the \Martingale{} transform.  
Let $S_i$ be the state of the original sketch after seeing
$(a_1,\ldots,a_i)$ and $P_{i+1} = \Pr(S_{i+1}\neq S_i \mid S_i, a_{i+1}\not\in\{a_1,\ldots,a_i\})$ be the probability
that it changes state upon seeing a \emph{new} element $a_{i+1}$.\footnote{These probabilities are over the choice of $h(a_{i+1})$, which,
in the \textsc{random oracle model}, is independent of all
other hash values.}  The state of the \Martingale{} sketch
is $(S_i,\hat{\lambda}_i)$.  Upon processing $a_{i+1}$ it becomes
$(S_{i+1},\hat{\lambda}_{i+1})$, where
\[
\hat{\lambda}_{i+1} = \hat{\lambda}_i + P_{i+1}^{-1}\cdot \Indicator{S_{i+1}\neq S_i}.
\]
Here $\Indicator{\mathcal{E}}$ is the indicator variable
for the event $\mathcal{E}$.  We assume 
the original sketch is insensitive to duplicates, 
so
\[
\E(\hat{\lambda}_{i+1}) =
\left\{\begin{array}{l@{\hspace*{1cm}}l}
\hat{\lambda}_i     & \mbox{when $a_{i+1}\in \{a_1,\ldots,a_i\}$ (and hence $S_{i+1}=S_i$)}\\
\hat{\lambda}_i+1   & \mbox{when $a_{i+1}\not\in \{a_1,\ldots,a_i\}$.}
\end{array}\right.
\]
Thus, with $\hat{\lambda}_0 = \lambda_0 = 0$,
$\hat{\lambda}_i$ is an unbiased estimator of the true cardinality
$\lambda_i=|\{a_1,\ldots,a_i\}|$ and
$(\hat{\lambda}_i - \lambda_i)_i$ is a \emph{martingale}.  The \Martingale-transformed sketch requires the same space, plus just $\log U$ bits to store the estimate $\hat{\lambda}$.

Cohen and Ting~\cite{Cohen15,Ting14} both proved that
\Martingale{} \MinCount{} has standard error $1/(2\sqrt{m})$
and $\MVP = (\log U)/2$.  They gave different estimates 
for the standard error of \Martingale{} \LogLog.
Ting's estimate is quite accurate, and 
tends to $\sqrt{\ln 2/m}$ as $m\rightarrow \infty$,
giving it an $\MVP = \ln 2\log\log U \approx 0.69\log\log U$.

\subsection{The Dartboard Model}

The~\emph{\textbf{dartboard model}}~\cite{PettieW21} 
is useful for describing cardinality sketches with a single,
uniform language.  The dartboard model is essentially the same as 
Ting's~\cite{Ting14} \emph{area cutting} process, but with
a specific, discrete cell partition and state space fixed in advance.

The \emph{dartboard} is the unit square $[0,1]^2$, partitioned 
into a set  $\mathscr{C}=\{c_0,\ldots,c_{|\mathscr{C}|-1}\}$ of \emph{cells}
of various sizes.
Every cell may be either \emph{occupied} 
or \emph{unoccupied};
the \emph{state} is the set of occupied cells
and the state space some $\mathscr{S}\subseteq 2^{\mathscr{C}}$.

We process a stream of elements 
one by one; when a \emph{new} element is encountered we throw a \emph{dart} uniformly at random 
at the dartboard and update the state in response.  
The relationship between the state and the dart distribution satisfies two rules:
\begin{description}
\item[(R1)] Every cell with at least one dart is occupied; occupied cells may contain no darts.

\item[(R2)] If a dart lands in an occupied cell, the state does not change.
\end{description}

As a consequence of (R1) and (R2), if a dart lands in an empty cell 
the state \emph{must} change, 
and occupied cells may never become unoccupied.  
Dart throwing is merely an intuitive way of visualizing the hash 
function.  
Base-$q$ \PCSA{} and \LogLog{} use the same 
cell partition 
but with different state spaces; 
see Figure~\ref{fig:PCSA-LogLog-dartboard}.

\begin{figure}
    \centering
    \begin{tabular}{c@{\hspace{.4cm}}c@{\hspace{.4cm}}c}
    \scalebox{.25}{\includegraphics{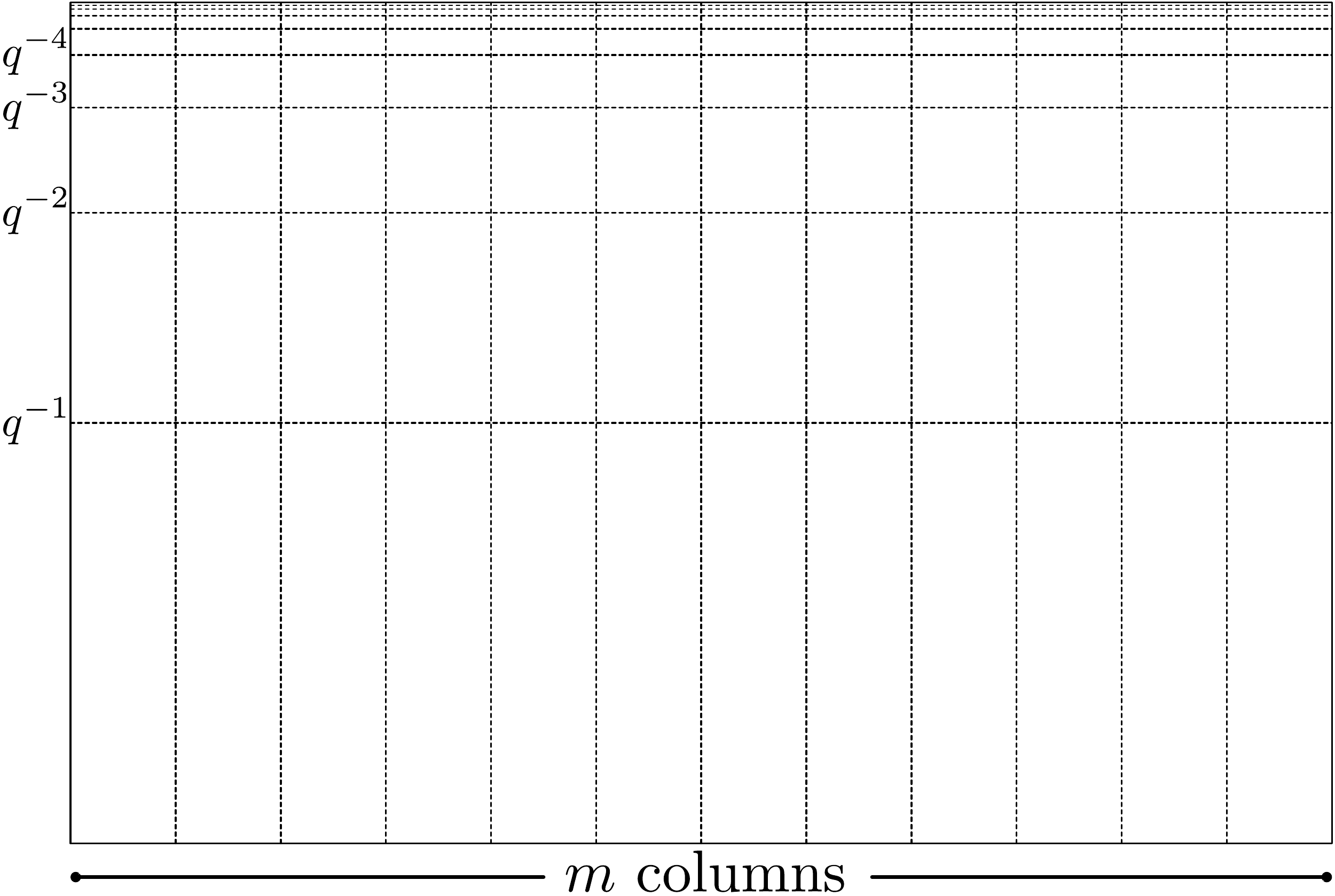}}
&
    \scalebox{.25}{\includegraphics{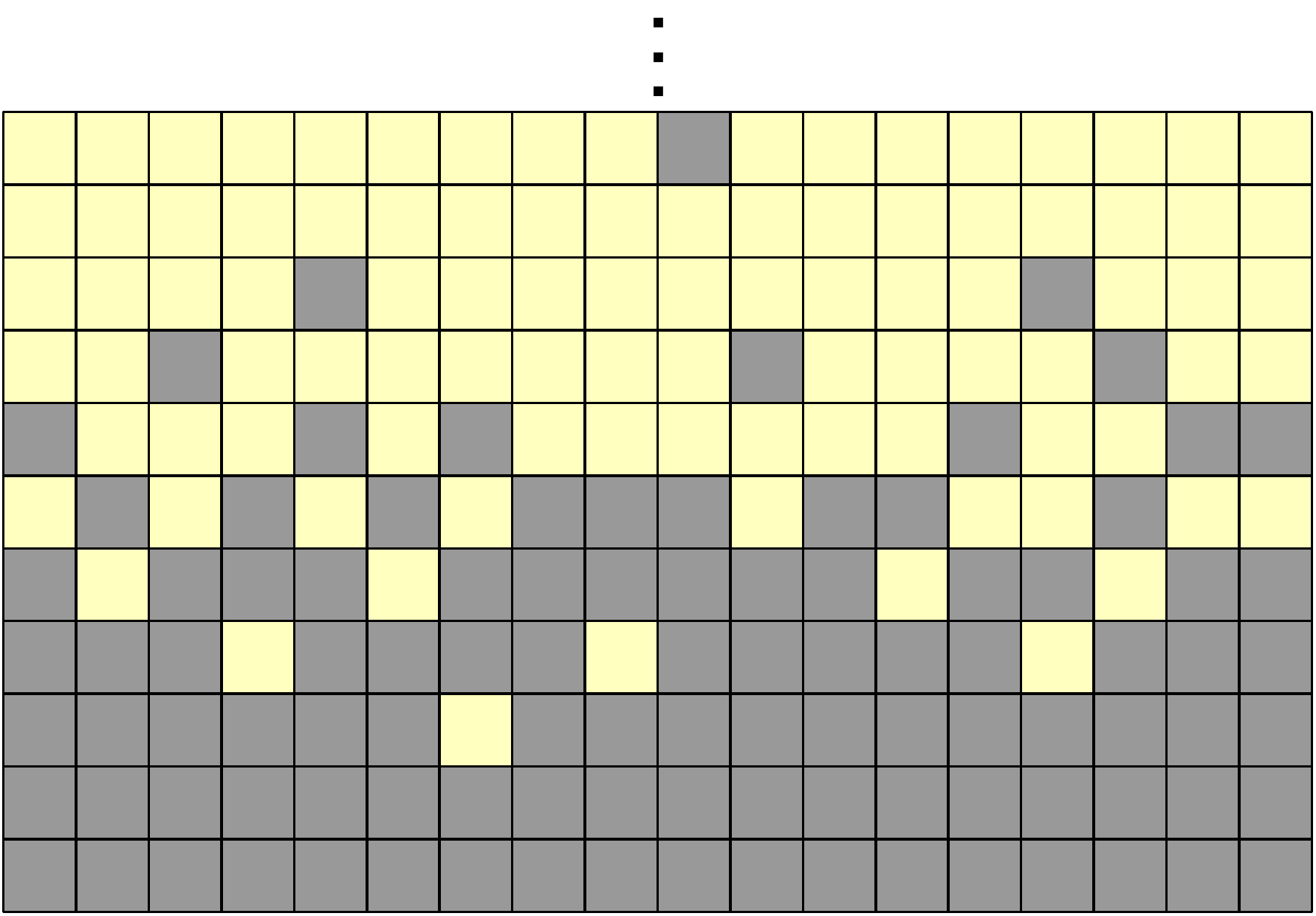}}  
    & \scalebox{.25}{\includegraphics{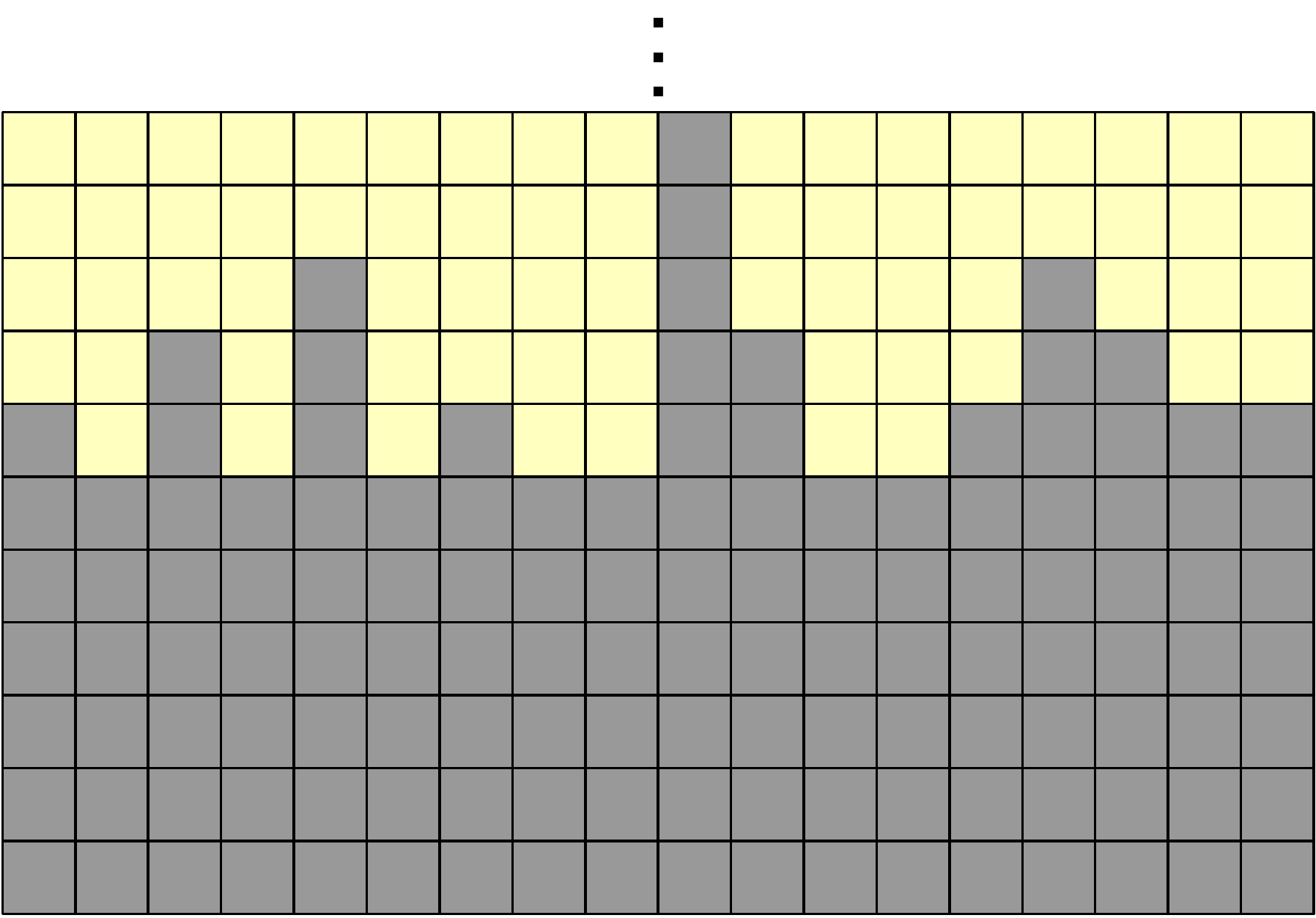}}\\
    &\\
    (a) & (b) & (c)
    \end{tabular}
    \caption{The unit square is partitioned into $m$ columns.  Each column
    is partitioned into cells.  Cell $j$ covers the vertical interval $[q^{-(j+1)},q^{-j})$.
    (b) The state of a $\PCSA$ sketch records precisely which cells
    contain a dart (gray); all others are empty (yellow).  (c) The state
    of the corresponding \textsf{LogLog} sketch.}
    \label{fig:PCSA-LogLog-dartboard}
\end{figure}

It was observed~\cite{PettieW21} that the dartboard model includes
all mergeable sketches, and some non-mergeable ones like \textsf{S-Bitmap}.
\textsf{Recordinality} and the \Martingale{} sketches obey rules (R1),(R2) but are not strictly dartboard sketches as they maintain some small state 
information ($\operatorname{cnt}$ or $\hat{\lambda}$) outside of 
the set of occupied cells.
Nonetheless, it is useful to speak of the \emph{dartboard part} of their state information.

\subsection{Linearizable Sketches}

The lower bound of~\cite{PettieW21} applies to \emph{linearizable} sketches,
a subset of mergeable sketches.  A sketch is called linearizable if it is 
possible to encode the occupied/unoccupied status of its cells in some fixed linear order $(c_0,\ldots,c_{\mathscr{C}-1})$, so whether $c_i$ is occupied only depends on the status of $c_0,\ldots,c_{i-1}$ and whether $c_i$ has been hit by a dart. (Thus, it is independent of $c_{i+1},\ldots,c_{\mathscr{C}-1}$.)  Specifically, let $Y_i,Z_i$ be the indicators for whether $c_i$ is occupied, and has been hit by a dart, respectively, and $\mathbf{Y}_i = (Y_0,\ldots,Y_i)$.  The state of the sketch is $\mathbf{Y}_{\mathscr{C}-1}$; it is called linearizable if there is some monotone function $\phi : \{0,1\}^*\rightarrow \{0,1\}$ such that 
\[
Y_i = Z_i \vee \phi(\mathbf{Y}_{i-1}).
\]
I.e., if $\phi(\mathbf{Y}_{i-1})=1$, $c_i$ is forced to be occupied and the state is forever independent of $Z_i$. 

\PCSA-type sketches~\cite{FlajoletM85,EstanVF06} are linearizable,
as are (\textsf{Hyper})\LogLog~\cite{FlajoletFGM07,DurandF03}, and all \MinCount, \textsf{MinHash}, and \textsf{Bottom}-$m$ type sketches~\cite{Cohen97,Broder97,Giroire09,ChassaingG06,Lumbroso10}.
It is very easy to engineer non-linearizable sketches; see~\cite{PettieW21}.
The open problem is whether this is ever a \emph{good idea} in 
terms of memory-variance performance.

\subsection{Organization}

In Section~\ref{sect:curtain} we 
introduce the \Curtain{} sketch, which
is a linearizable (hence mergeable) sketch in the dartboard model.
In Section~\ref{sect:mtg_est} we prove some general
theorems on the bias and asymptotic 
relative variance of \Martingale-type sketches,
and in Section~\ref{sect:analysis} we apply this framework
to bound the limiting $\MVP$
of \Martingale{} \PCSA, \Martingale{} \fishmonger, and \Martingale{} \Curtain.

In Section~\ref{sect:optimality} we prove some 
results on the optimality
of the \Martingale{} transform itself,
and that \Martingale{} \fishmonger{}
has the lowest variance among those based 
on linearizable sketches.

Section~\ref{sect:experiments} presents some experimental
findings that demonstrate that the
conclusions drawn from the asymptotic 
analysis of \Martingale{} sketches are extremely accurate in the pre-asymptotic 
regime as well, and that \Martingale{} \Curtain{} has lower variance than
\Martingale{} \LogLog.

\section{The Curtain Sketch}\label{sect:curtain}

\paragraph{Design Philosophy.} Our goal is to strike a nice balance between the simplicity and time-efficiency of (\textsf{Hyper})\LogLog, and the superior information-theoretic efficiency of \PCSA, 
which can only be fully realized 
under extreme (and time-inefficient) compression to its entropy bound~\cite{PettieW21,Lang17}.
Informally, if we are dedicating at least 1 bit to encode the status of a cell, the \emph{best} cells to encode have mass $\Theta(\lambda^{-1})$ and we should design a sketch that maximizes the number of such cells encoded.
\medskip

We assume the dartboard is partitioned into $m$ columns; define
$\Cell(j,i)$ to be the cell in column $i$ covering the vertical interval
$[q^{-(j+1)},q^{-j})$.
In a $\PCSA$ sketch, the occupied cells are precisely those with at least
one dart.  In \textsf{LogLog}, the occupied cells in each column
are contiguous, extending to the highest cell containing a dart.
In Figure~\ref{fig:PCSA-LogLog-dartboard}, cells are drawn with 
uniform sizes for clarity.

Consider the vector $v = (g_0,g_1,\ldots,g_{m-1})$
where $\Cell(g_i,i)$ is the highest occupied cell
in \textsf{LogLog}/$\PCSA$.
The \emph{curtain} of $v$ w.r.t.~allowable 
offsets $\mathscr{O}$ is a vector 
$v_{\operatorname{curt}} = (\hat{g}_0,\hat{g}_1,\ldots,\hat{g}_{m-1})$
such that
(i) $\forall i\in [1,m-1].\, \hat{g}_i - \hat{g}_{i-1}\in \mathscr{O}$,
and
(ii) $v_{\operatorname{curt}}$ is the minimal such vector
dominating $v$, i.e., $\forall i.\, \hat{g}_i \geq g_i$.
Although we have described $v_{\operatorname{curt}}$ as a 
function of $v$, it is clearly possible to maintain $v_{\operatorname{curt}}$
as darts are thrown, without knowing $v$.

We have an interest in $|\mathscr{O}|$ being a power
of 2 so that curtain vectors may be encoded efficiently,
as a series of offsets.
On the other hand, it is most efficient if 
$\mathscr{O}$ is symmetric around zero.  
For these reasons,
we use a base-$q$ ``sawtooth'' cell partition of the dartboard;
see Figure~\ref{fig:Curtain-example}.  
Henceforth $\Cell(j,i)$ is defined as usual, except
$j$ is an integer when $i$ is even and a half-integer 
when $i$ is odd.  Then the allowable offsets are 
$\mathscr{O}_a = \{-(a-1/2), -(a-3/2),\ldots, -1/2, 1/2, \ldots, a-3/2,a-1/2\}$, for some $a$ that is a power of 2.

\begin{figure}
    \centering
    \begin{tabular}{c@{\hspace{0.4cm}}c@{\hspace{0.4cm}}c}
    \scalebox{.20}{\includegraphics{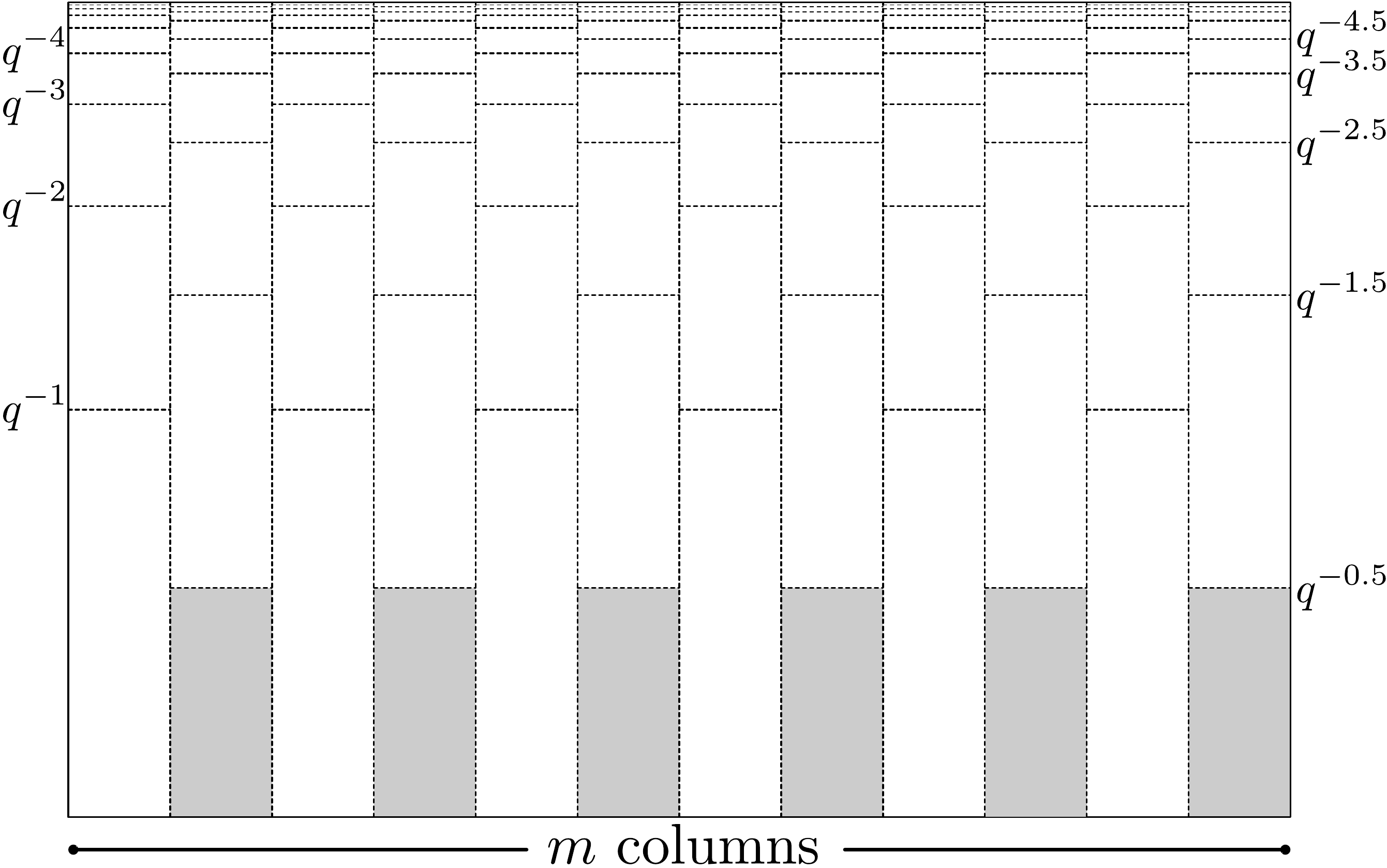}}
    &
    \scalebox{.25}{\includegraphics{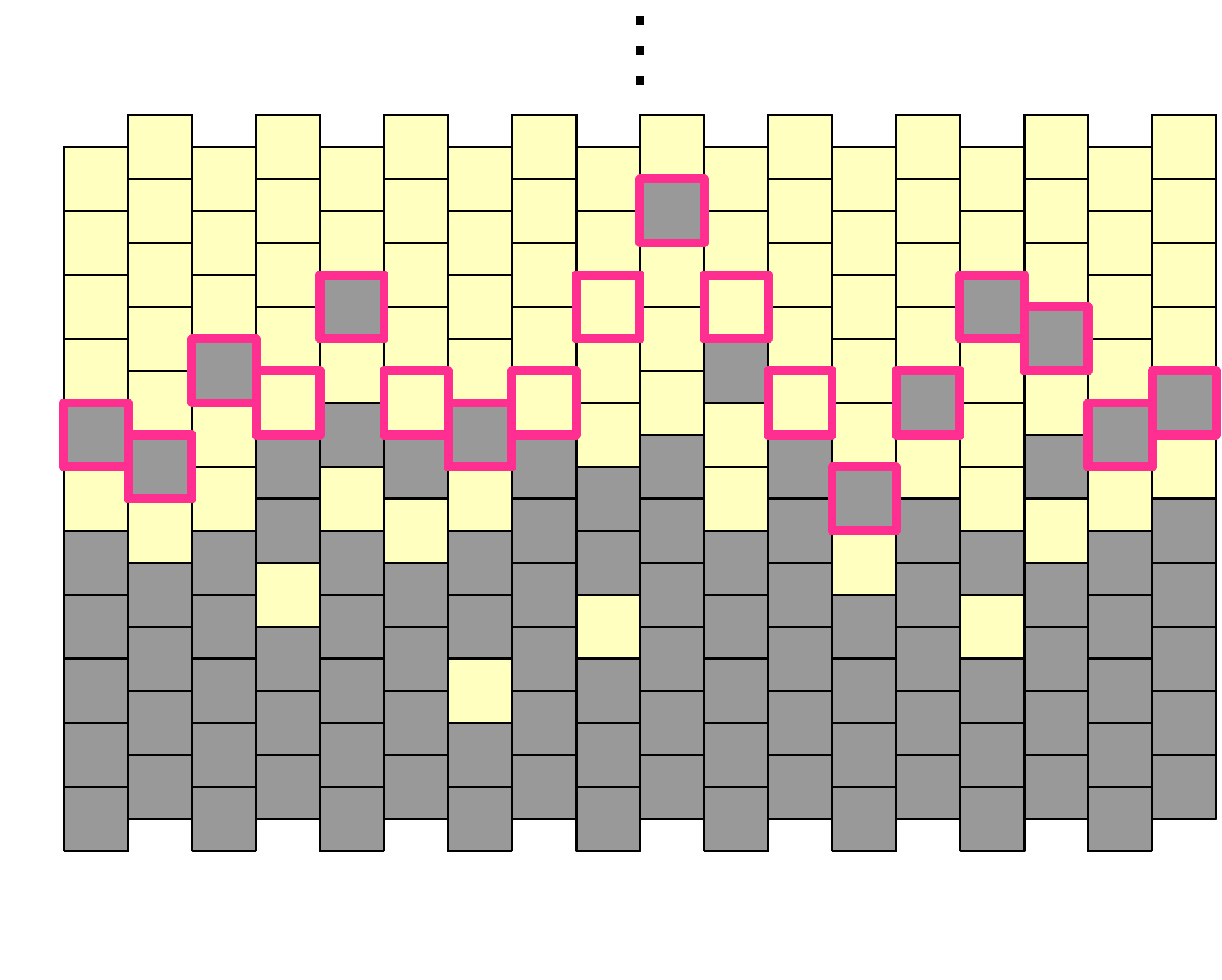}}
    &
    \scalebox{.25}{\includegraphics{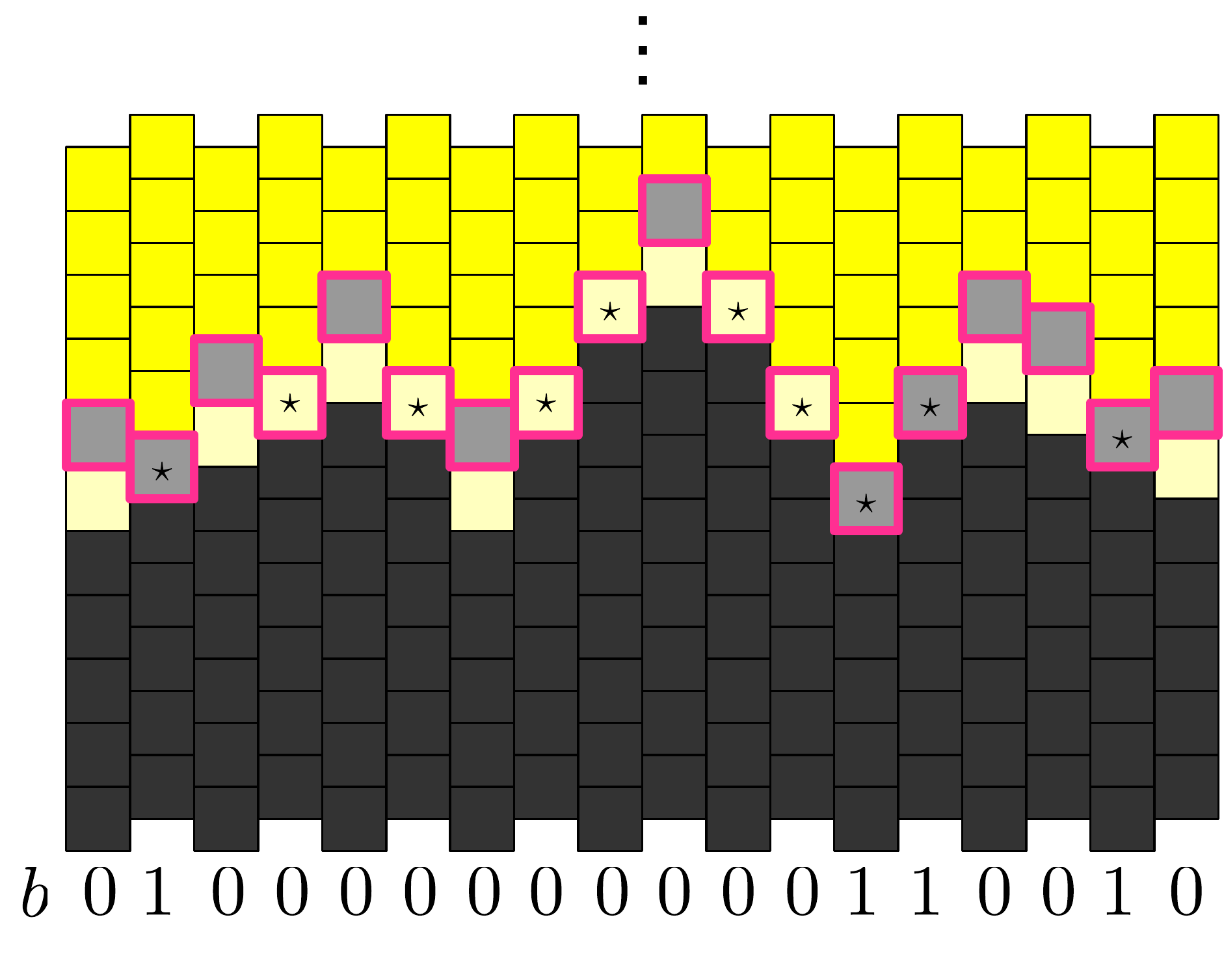}}\\
    (a) & (b) & (c)\\
    \end{tabular}
    \caption{(a) The base-$q$ ``sawtooth'' cell partition.
    (b) and (c) depict a \Curtain{} sketch w.r.t. $\mathscr{O}=\{-3/2,-1/2,1/2,3/2\}$ and $h=1$.
    (b) Gray cells contain at least one dart;
    light yellow cells contain none.  
    The curtain $v_{\operatorname{curt}} = (\hat{g}_i)$ 
    is highlighted with a pink boundary.
    (c) Columns that are in \emph{tension}
    have a $\star$ in their curtain cell.
    All \emph{dark} gray cells are occupied and 
    all \emph{dark} yellow cells are free according to Rule 3.
    All other cells are occupied/free (light gray, light yellow) according to 
    Rules 1 and 2.}
    \label{fig:Curtain-example}
\end{figure}

Let $\Cell(g_i,i)$ is the highest cell containing
a dart in column $i$ in the \emph{sawtooth} cell partition
and $v_{\operatorname{curt}} = (\hat{g}_i)$ be the curtain vector of
$v=(g_i)$ w.r.t.~offsets $\mathscr{O} = \mathscr{O}_a$.
We say column $i$ is \emph{in tension} if $(\cdots,\hat{g}_{i-1},\hat{g}_i-1,\hat{g}_{i+1},\cdots)$
is not a valid curtain, i.e., if $\hat{g}_{i} - \hat{g}_{i-1} = \min(\mathscr{O})$ 
or $\hat{g}_{i+1} - \hat{g}_i = \max(\mathscr{O})$.
In particular, if column $i$ is \emph{not} in tension, 
then $\Cell(\hat{g}_i,i)$ must contain at least one dart,
for if it contained no darts the curtain would be dropped 
to $\hat{g}_i-1$ at column $i$.
However, if column $i$ \emph{is} in tension, 
then $\Cell(\hat{g}_i,i)$ might not contain a dart.

\medskip

The \Curtain{} sketch encodes 
$v_{\operatorname{curt}} = (\hat{g}_i)$ w.r.t.~the base-$q$
sawtooth cell partition and offsets $\mathscr{O}_a$,
and a bit-array $b = \{0,1\}^{h\times m}$.
This sketch designates each cell 
\emph{occupied} or \emph{free} as follows.
\begin{description}
    \item[Rule 1.] If column $i$ is not in tension then
    $\Cell(\hat{g}_i,i)$ is occupied, and $b(\cdot,i)$ encodes
    the status of the $h$ cells below the curtain, i.e.,
    $\Cell(\hat{g}_i - (j+1),i)$ is occupied iff $b(j,i)=1$, 
    $j\in \{0,\ldots,h-1\}$.
    \item[Rule 2.] If column $i$ is in tension, then 
    $\Cell(\hat{g}_i - j, i)$ is occupied iff $b(j,i)=1$,
    $j\in \{0,\ldots,h-1\}$.
    \item[Rule 3.] Every cell above the curtain is free 
    ($\Cell(\hat{g}_i+j,i)$, when $j\ge 1$)
    and all remaining cells are occupied.
\end{description}

Figure~\ref{fig:Curtain-example} gives an example of a \Curtain{}
sketch, with $\mathscr{O} = \{-3/2,-1/2,1/2,3/2\}$
and $h=1$.  (The base $q$ of the cell partition is unspecified
in this example.)

\begin{theorem}\label{thm:martingale-curtain}
Consider the \Martingale{} \Curtain{} sketch with parameters $q,a,h$
(base $q$, $\mathscr{O}_a = \{-(a-1/2),\ldots,a-1/2\}$, and $b\in \{0,1\}^{h\times m}$),
and let $\hat{\lambda}$ be its estimate of the true cardinality $\lambda$.
\begin{enumerate}
    \item $\hat{\lambda}$ is an unbiased estimate of $\lambda$.
    \item  The relative variance of $\hat{\lambda}$ is:
    \[
    \frac{1}{\lambda^2}\Var(\hat{\lambda} \mid \lambda) 
    = \frac{(1+o_{\lambda/m}(1)+o_m(1))q\ln q}{2m(q-1)}\left(\frac{q-1}{q}+\frac{2}{q^{h}(q^{a-1/2}-1)}+\frac{1}{q^{h+1}}\right),
    \]
    As a result, the limiting $\MVP$ of \Martingale{} \Curtain{} is 
    \[\MVP = (\log_2(2a)+h)\times \frac{q\ln q}{2(q-1)} \left(\frac{q-1}{q}+\frac{2}{q^{h}(q^{a-1/2}-1)}+\frac{1}{q^{h+1}}\right).\]
  \end{enumerate}
\end{theorem}
\begin{proof}
Follows from Theorems \ref{thm:kappa} and \ref{thm:curtain-kappa}.
\end{proof}

Here $o_{\lambda/m}(1)$ and $o_m(1)$ are terms that go to zero as $m$
and $\lambda/m$ get large.
Recall that for practical reasons we want to parameterize
Theorem~\ref{thm:martingale-curtain} with $a$ a power of 2
and $h$ an integer, but it is realistic to set $q>1$ to 
be any real.  Given these constraints, the optimal setting
is $q=2.91$, $a=2$, and $h=1$, exactly as in the example in Figure~\ref{fig:Curtain-example}.
This uses $\log\log U + 3(m-1)$ bits to store the sketch proper, 
$\log U$ bits\footnote{It is fine to store an approximation $\tilde{\lambda}$ of $\hat{\lambda}$
with $O(\log m)$ bits of precision.}
to store $\hat{\lambda}$, 
and achieves a limiting $\MVP \approx 2.31$.
In other words, to achieve a standard error $1/\sqrt{b}$, 
we need about $2.31b$ bits.

\paragraph{Implementation Considerations.}
We encode a curtain $(\hat{g}_0,\hat{g}_1,\ldots,\hat{g}_{m-1})$ as 
$\hat{g}_0$ and 
an offset vector $(o_1,o_2,\ldots,o_{m-1})$, $o_i=\hat{g}_i-\hat{g}_{i-1}$,
where $\hat{g}_0$ takes $\log_2\log_q U \leq 6$ 
bits and $o_i$ takes $\log_2|\mathscr{O}|=\log_2(2a)$ bits.
Clearly, to evaluate $\hat{g}_i$ we need to compute the prefix 
sum $\hat{g}_0 + \sum_{i'\leq i} o_{i'}$.

\begin{lemma}\label{lem:packed-word}
Let $(x_0,\ldots,x_{\ell-1})$ be a vector of $t$-bit unsigned integers
packed into $\ceil{t\ell/w}$ words, where each word has
$w=\Omega(\log(t\ell))$ bits.
The prefix sum $\sum_{j\in [0,i]} x_j$ can be evaluated 
in $O(t\ell/w + \log w)$ time.
\end{lemma}
\begin{proof}
W.l.o.g.~we can assume $i=\ell-1$, so the task is to sum
the entire list.  
In $O(\ceil{(t\ell)/w})$ time we can 
halve the number of summands, 
by masking out the odd and even summands and adding these vectors
together.  After halving twice in this way, we have a vector of 
$\ell/4$ $(t+2)$-bit integers, each allocated $4t$ bits.
At this point we can halve the number of words by adding
the $(2i+1)$th word to the $2i$th word.
Thus, if $T_w(\ell,t)$ is the time needed to solve this problem,
$T_w(\ell,t) = T_w(\ell/8,4t) + O(\ceil{(t\ell)/w})$, 
which is $O((t\ell)/w + \log w)$.
\end{proof}

In our context $t=\log_2(2a)=2$, so even if $m$ is a medium-size
constant, say at most 256 or 512, we only have to do prefix sums over 
8 or 16 consecutive 64-bit words.  
If $m$ is much larger then it would be prudent
to partition the dartboard into $m/c$ independent curtains, 
each with $c = 256$ or 512 columns.
This keeps the update time independent of $m$ and 
increases the space overhead negligibly.

We began this section by highlighting the design philosophy,
which emphasizes conceptual simplicity and efficiency. 
Our encoding uses fixed-length codes for the offsets,
and can be decoded very efficiently by exploiting bit-wise
operations and word-level parallelism.  That said, we are 
mainly interested in analyzing 
the \emph{theoretical} performance of sketches,
and will not attempt an exhaustive experimental evaluation in this work.

\section{Foundations of the \Martingale{} Transform}\label{sect:mtg_est}

In this section we present a simple framework for analyzing the limiting
variance of \Martingale{} sketches, which is strongly influenced 
by Ting's~\cite{Ting14} work.
Theorem~\ref{thm:me} gives simple unbiased estimators for the cardinality
and the variance of the the cardinality estimator.  The upshot of Theorem~\ref{thm:me}
is that to analyze the variance of the estimator, we only need to bound $\E(P_k^{-1})$,
where $P_k$ is the probability the $k$th distinct element changes the sketch.
Theorem~\ref{thm:kappa} further shows that for sketches composed of $m$
subsketches (like \Curtain, \textsf{HyperLogLog}, and $\PCSA$), 
the limiting variance tends to $\frac{1}{2\kappa m}$, 
where $\kappa$ is a constant that depends on the sketch scheme.  
Section~\ref{sect:analysis} analyzes the constant $\kappa$
for each of $\PCSA$, \textsf{LogLog}, and \Curtain.
Using results of~\cite{PettieW21} on the entropy of \PCSA{} we can calculate the limiting $\MVP$ of
\PCSA, \LogLog, \Curtain, and \fishmonger.

\subsection{Martingale Estimators and Retrospective Variance}\label{sect:V-E}

Consider an arbitrary sketch with state space $\mathcal{S}$.  We assume
the sketch state does not change upon seeing duplicated elements, hence
it suffices to consider streams of \emph{distinct} elements.
We model the evolution of the sketch as a Markov chain
$(S_k)_{k\geq 0} \in \mathcal{S}^*$, where $S_k$ is 
the state after seeing $k$ \emph{distinct} elements.
Define $P_k = \Pr(S_{k} \neq S_{k-1}  \mid S_{k-1})$
to be the \emph{state changing probability}, which depends only on $S_{k-1}$.
In the dartboard terminology $P_k$
is the total size of all unoccupied cells in $S_{k-1}$.

\begin{definition}
Let $\Indicator{\mathcal{E}}$
be the indicator variable for event $\mathcal{E}$. For any $\lambda\geq 0$, define:
\begin{align*}
    E_\lambda   &= \sum_{k=1}^\lambda\Indicator{S_k\neq S_{k-1}}\cdot \frac{1}{P_k}, & \mbox{the martingale estimator,}\\
\text{ and } V_\lambda &= \sum_{k=1}^\lambda\Indicator{S_k\neq S_{k-1}}\cdot \frac{1-P_k}{P_k^2}, & \mbox{the ``retrospective'' variance.}
\end{align*}
Note that $E_0=V_0=0$.
\end{definition}

The \Martingale{} transform of this sketch stores $\hat{\lambda} = E_\lambda$ in one machine word
and returns it as a cardinality estimate.  It can also store $V_\lambda$ in one machine word as well.
Theorem~\ref{thm:me} shows that the retrospective variance $V_\lambda$ is a good running estimate
of the empirical squared error $(E_\lambda-\lambda)^2$.

\begin{theorem}\label{thm:me}
The martingale estimator $E_\lambda$ is an unbiased estimator of $\lambda$ and the retrospective variance $V_\lambda$ is an unbiased estimator of $\Var(E_\lambda)$. Specifically, we have,
\begin{align*}
    \E(E_\lambda) =\lambda,\text{ and } \Var (E_\lambda)=\E(V_\lambda) = \sum_{k=1}^\lambda 
    \E\paren{\frac{1}{P_k}}-\lambda.
\end{align*}
\end{theorem}

\begin{rem}
Theorem~\ref{thm:me} contradicts Ting's claim~\cite{Ting14}, that $V_\lambda$
is unbiased \emph{only at ``jump'' times}, i.e., those $\lambda$ for which $S_\lambda\neq S_{\lambda-1}$,
and therefore inadequate to estimate the variance.
In order to correct for this, Ting introduced a Bayesian method for estimating the time
that has passed since the last jump time.  The reason for thinking that jump times are different
is actually quite natural.  Suppose we record the list of \underline{\emph{distinct}} states $s_0,\ldots,s_k$
encountered while inserting $\lambda$ elements, $\lambda$ being unknown, and let 
$p_i$ be the probability of changing from $s_i$ to some other state.  
The amount of time spent in state $s_i$ is a geometric random variable with mean $p_i^{-1}$
and variance $(1-p_i)/p_i^2$.  Furthermore, these waiting times are independent.  Thus,
$\sum_{i\in [0,k)} p_i^{-1}$ and $\sum_{i\in [0,k)} (1-p_i^{-1})/p_i^2$
are unbiased estimates of the cardinality $\lambda'$ and squared error \emph{upon entering state $s_k$}.
These exactly correspond to $E_\lambda$ and $V_\lambda$, but they \emph{should} be biased since
they do not take into account the $\lambda-\lambda'$ elements that had no effect on $s_k$.
As Theorem~\ref{thm:me} shows, this is a mathematical optical illusion.  The history
is a random variable, and although the last $\lambda-\lambda'$ elements did not change the state,
\emph{they could have}, which would have altered the observed history $s_0,\ldots,s_k$
and hence the estimates
$E_\lambda$ and $V_\lambda$.
\end{rem}

\begin{proof}[Proof of Theorem~\ref{thm:me}]
Note that $P_k$ is a function of $S_{k-1}$. 
By the linearity of expectation and the law of total expectation, we have
\begin{align*}
    \E(E_k) &= \E(\E(E_k\mid S_{k-1}))=\E\paren{\E(E_{k-1} \mid S_{k-1})+\E\paren{\Indicator{S_k\neq S_{k-1}}\cdot \frac{1}{P_k} \;\middle|\; S_{k-1}}}\\
    &= \E(E_{k-1})+1=\E(E_{k-2})+2=\ldots=\E(E_{0})+k=k.
\intertext{and} 
    \E(V_k) &= \E(\E(V_k\mid S_{k-1}))=\E\paren{\E(V_{k-1}\mid S_{k-1})+\E\paren{\Indicator{S_k\neq S_{k-1}}\cdot\frac{1-P_k}{P_k^2} \;\middle|\; S_{k-1}}}\\
    &=\E(V_{k-1})+\E\paren{\frac{1-P_k}{P_k}}=\E(V_{k-2})+\E\paren{\frac{1-P_k}{P_k}}+\E\paren{\frac{1-P_{k-1}}{P_{k-1}}}=\ldots\\
    &= \E(V_{0})+\sum_{i=1}^k \E\paren{\frac{1-P_i}{P_i}} =\sum_{i=1}^k \E\paren{\frac{1}{P_i}}-k.
\end{align*}
For the variance,  we have
\begin{align*}
    \Var(E_\lambda) &= \E(E_\lambda^2)-(\E(E_\lambda))^2=\E(E_\lambda^2)-\lambda^2.
\end{align*}

Note that
\begin{align*}
    \E(E_{k}^2 \mid S_{k-1}) &= \E\paren{\paren{E_{k-1}+\Indicator{S_k\neq S_{k-1}}\cdot \frac{1}{P_k}}^2 \;\middle|\; S_{k-1}}\\ 
    &= E_{k-1}^2+2\frac{E_{k-1}}{P_k}\cdot \E\paren{\Indicator{S_k\neq S_{k-1}} \;\middle|\; S_{k-1}}+\frac{1}{P_k^2}\cdot \E\paren{\Indicator{S_k\neq S_{k-1}}^2 \;\middle|\; S_{k-1}}\\
    &= E_{k-1}^2+2E_{k-1}+\frac{1}{P_k}.
\intertext{Then by the law of total expectation and the linearity of expectation, we have}
    \E\paren{E_k^2} &=\E\paren{\E\paren{E_k^2\mid S_{k-1}}}=\E\paren{E_{k-1}^2+2E_{k-1}+\frac{1}{P_k}}
    \;=\; \E\paren{E_{k-1}^2}+2(k-1)+\E\paren{\frac{1}{P_k}}.
\intertext{From this recurrence relation, we have}
    \E\paren{E_\lambda^2}&=\E\paren{E_0^2}+2\sum_{k=1}^{\lambda}(k-1)+\sum_{k=1}^{\lambda}\E\paren{\frac{1}{P_k}}
    \;=\; \sum_{k=1}^{\lambda}\E\paren{\frac{1}{P_k}}+\lambda(\lambda-1).
\end{align*}


We conclude that
\begin{align*}
    \Var(E_\lambda)
    \;=\; \sum_{k=1}^{\lambda}\E\paren{\frac{1}{P_k}}+\lambda(\lambda-1) - \lambda^2 
    \;=\; \sum_{k=1}^{\lambda}\E\paren{\frac{1}{P_k}} - \lambda 
    \;=\; \E(V_\lambda).
\end{align*}
\end{proof}

\subsection{Asymptotic Relative Variance}\label{sect:ARV}

\subsubsection{The ARV Factor}

We consider classes of sketches composed of $m$ \emph{subsketches}, which controls the size and variance. 
In \textsf{LogLog}, $\PCSA$, and \Curtain{} these subsketches are the $m$ columns.  When considering a sketch with $m$ subsketches, instead of using $\lambda$ as the total number of  insertions, we always use $\lambda$ to denote the number of insertions \emph{per subsketch} and therefore the total number of insertions is $\lambda m$.
We care about the \emph{asymptotic relative variance} (ARV) as $m$ and $\lambda$ both go to infinity 
(defined below).  A reasonable sketch should have relative variance $O(1/m)$.  Informally, the ARV factor is just
the leading constant of this expression.

\begin{definition}[ARV factor]\label{def:ARV}
Consider a class of sketches whose size is parameterized by $m$. For any $k\geq 0$, define 
$P_{m,k}$ to be the probability the sketch changes state upon the $k$th insertion and $E_{m,k}$ 
the martingale estimator.
The \emph{ARV factor} of this class of sketches is defined as
\begin{align}
    \lim_{\lambda\to\infty} \lim_{m\to\infty} m\cdot  \frac{\Var(E_{m,\lambda m})}{(\lambda m)^2}.\label{eq:arvf}
\end{align}
\end{definition}

\subsubsection{Scale-Invariance and the Constant $\kappa$}

Few sketches have \emph{strictly} well-defined ARV factors. In \Martingale{} \textsf{LogLog}, for example, the quantity 
$\left(\lim_{m\to\infty} m \frac{\Var(E_{m,\lambda m})}{(\lambda m)^2}\right)$
is not constant, but 
periodic in $\log_2 \lambda$; it does not converge as $\lambda\to\infty$.
We explain how to fix this issue using \emph{smoothing} in Section~\ref{sect:smoothing}.
\emph{Scale-invariant} sketches must have well-defined ARV factors.

\begin{definition}[scale-invariance and constant $\kappa$]\label{def:scale-invariant}
A combined sketch is \emph{scale-invariant} if 
\begin{enumerate}
    \item 
For any $\lambda$, there exists a constant $\kappa_\lambda$ such that $
 \lambda \cdot P_{m, \lambda m}$ converges to $\kappa_\lambda$ almost surely as $m\to\infty$.
\item 
The limit of $\kappa_\lambda$ as $\lambda\to\infty$ exists, and $\kappa \bydef \lim_{\lambda\to\infty} \kappa_\lambda$.
\end{enumerate}
The constant of a sketch $A$ is denoted as $\kappa_A$, where the subscript $A$ is often dropped when the context is clear.
\end{definition}

The next theorem proves that under mild regularity conditions, all scale-invariant sketches have well 
defined ARV factors and there is a direct relation between the ARV factor and the 
constant $\kappa$.  

\begin{theorem}[ARV factor of a scale-invariant sketch]\label{thm:kappa} Consider a sketching scheme satisfying the following properties.
\begin{enumerate}
    \item It is scale-invariant with constant $\kappa$.
    \item For any $\lambda > 0$, the limit operator and the expectation operator of $\{\frac{1}{P_{m,\lambda m}}\}_{m}$ can be interchanged.
\end{enumerate}
 Then the ARV factor of the sketch exists and 
 equals $\frac{1}{2\kappa}$.
\end{theorem}
\begin{proof}
First note that, by the assumptions, we have that
\begin{align*}
    \lim_{m\to\infty} \E\paren{\frac{1}{P_{m,\lambda m}}} &=  \E\paren{\lim_{m\to\infty}\frac{1}{P_{m,\lambda m}}} = \E\paren{\frac{\lambda}{\kappa_\lambda}}= \frac{\lambda}{\kappa_\lambda}.
\end{align*}
Also note that since $P_{m,k}$ are non-increasing as $k$ increases, by simple coupling argument, we see that for any $k\leq k'$, $\mathbb{E}\paren{1/P_{m,k}}\leq \mathbb{E}\paren{1/P_{m,k'}}$ and $\frac{\kappa_k}{k}\geq \frac{\kappa_{k'}}{k'}$.

Fix $\lambda>0$, we have, by Theorem \ref{thm:me},
\begin{align}
    \lim_{m\to\infty}\frac{1}{\lambda^2m}\Var(E_{m,\lambda m})&= \lim_{m\to\infty} \left(\frac{1}{\lambda^2m} \sum_{k=1}^{\lambda m}\E\paren{ \frac{1}{P_{m,k}}}-\frac{1}{\lambda}\right)\nonumber\\
    &=  \lim_{m\to\infty} \frac{1}{\lambda^2m} \sum_{i=0}^{\lambda-1} \sum_{j=1}^{m}
    \E\paren{\frac{1}{P_{m,im+j}}} -\frac{1}{\lambda} \label{eq:mve}
\end{align}
Since for any $j\in[1,m]$, $ \E\paren{\frac{1}{P_{m,im+j}}} \leq \E\paren{\frac{1}{P_{m,(i+1)m}}}$, we have
\begin{align*}
    \lim_{m\to\infty}\frac{1}{\lambda^2m}\Var(E_{m,\lambda m})\leq &  \lim_{m\to\infty} \frac{1}{\lambda^2m} \sum_{i=0}^{\lambda-1}  \sum_{j=1}^{m}\E\paren{\frac{1}{P_{m,(i+1)m}}} -\frac{1}{\lambda}\\
    &=    \frac{1}{\lambda^2} \sum_{i=0}^{\lambda-1}  \lim_{m\to\infty} \E\paren{\frac{1}{P_{m,(i+1)m}}} -\frac{1}{\lambda}\\
    &=    \frac{1}{\lambda^2} \sum_{i=0}^{\lambda-1}   \frac{i+1}{\kappa_{i+1}} -\frac{1}{\lambda},
\end{align*}

Denote the ARV factor as $v$.
Fix $W>0$.  Note that for any $i\in[0,\lambda/W-1]$, $\frac{k\lambda/W+i+1}{\kappa_{k\lambda/W+i+1}}\leq \frac{(k+1)\lambda/W}{\kappa_{(k+1)\lambda/W}}$. 
\begin{align}
   v \leq & \lim_{\lambda\to\infty}\left(\frac{1}{\lambda^2} \sum_{i=0}^{\lambda-1}  \frac{i+1}{\kappa_{i+1}} -\frac{1}{\lambda}\right)
    =\lim_{\lambda\to\infty}\left(\frac{1}{\lambda^2} \sum_{k=0}^{W-1}\sum_{i=0}^{\lambda/W-1}   \frac{k\lambda/W+i+1}{\kappa_{k\lambda/W+i+1}} \right)\nonumber\\
    \leq&\lim_{\lambda\to\infty}\left(\frac{1}{\lambda^2} \sum_{k=0}^{W-1}\sum_{i=0}^{\lambda/W-1}   \frac{(k+1)\lambda/W}{\kappa_{(k+1)\lambda/W}}\right)
    =\frac{1}{W^2} \sum_{k=0}^{W-1}  \lim_{\lambda\to\infty}\frac{k+1}{\kappa_{(k+1)\lambda/W}} \nonumber\\
    \intertext{note that $\lim_{\lambda\to\infty}\kappa_{(k+1)\lambda/W}=\kappa$ by the definition of scale-invariance,}
    &=\frac{1}{W^2} \sum_{k=0}^{W-1} \frac{k+1}{\kappa}= \frac{1}{2\kappa} \frac{W(W+1)}{W^2}. \label{eq:v_upper}
\end{align}

On the other hand, we can  bound it from below similarly. We will only outline the key steps since it is almost identical to the previous one. Note that for any $j\in[1,m]$, $ \E\paren{\frac{1}{P_{m,im+j}}} \geq \E\paren{\frac{1}{P_{m,im}}}$. Using this inequality in (\ref{eq:mve}), we have
\begin{align*}
    \lim_{m\to\infty}\frac{1}{\lambda^2m}\Var(E_{m,\lambda m})\geq &  \lim_{m\to\infty} \frac{1}{\lambda^2m} \sum_{i=0}^{\lambda-1}  \sum_{j=1}^{m}\E\paren{\frac{1}{P_{m,im}}} -\frac{1}{\lambda}=    \frac{1}{\lambda^2} \sum_{i=0}^{\lambda-1}   \frac{i}{\kappa_i} -\frac{1}{\lambda}.
\end{align*}
Similarly, we have
\begin{align}
    v\geq &\lim_{\lambda\to\infty}\left(\frac{1}{\lambda^2} \sum_{k=0}^{W-1}\sum_{i=0}^{\lambda/(W)-1}   \frac{k\lambda/W+i}{\kappa_{k\lambda/W+i}} \right)
    \geq\lim_{\lambda\to\infty}\left(\frac{1}{\lambda^2} \sum_{k=0}^{W-1}\sum_{i=0}^{\lambda/W-1}  \frac{k\lambda/W}{\kappa_{k\lambda/W}} \right) \nonumber\\
    &=\frac{1}{W^2} \sum_{k=0}^{W-1} \frac{k}{\kappa}= \frac{1}{2\kappa} \frac{W(W-1)}{W^2}. \label{eq:v_lower}
\end{align}

Thus by combining (\ref{eq:v_upper}) and (\ref{eq:v_lower}), we have
\begin{align*}
   \frac{1}{2\kappa} \frac{W(W-1)}{W^2}\leq v \leq \frac{1}{2\kappa} \frac{W(W+1)}{W^2}.
\end{align*}
Since the choice of $W$ is arbitrary, we conclude that the ARV factor $v$ is well-defined and  $v=\frac{1}{2\kappa}$.
\end{proof}

The constant $\kappa$ together with Theorem \ref{thm:kappa} is useful in that it gives a simple 
and systematic way to evaluate the asymptotic performance of a well behaved (scale-invariant) sketch scheme. 

$\textsf{MinCount}$~\cite{Giroire09,ChassaingG06,Lumbroso10} is an example
of a scale-invariant sketch.  The function $h(a) = (i,v) \in [m]\times [0,1]$ 
is interpreted as a pair containing a bucket index and a real hash value.  
A $(k,m)$-\textsf{MinCount} sketch stores the smallest $k$ hash values in each bucket. 
\begin{theorem}
$(k,m)$-\textsf{MinCount} is scale-invariant and $\kappa_{(k,m)\textsf{-MinCount}} = k$.
\end{theorem}
\begin{proof}
When a total of $\lambda m$ elements are inserted to the combined sketch, each subsketch 
receives $(1+o(1))\lambda$ elements as $\lambda \to \infty$. Since we only care the asymptotic behavior, 
we assume for simplicity that each subsketch receives exactly $\lambda$ elements.

Let $P_\lambda^{(i)}$ be the probability that the sketch of 
the $i$th bucket changes after the $\lambda$th element is
thrown into the $i$th bucket.
Then by definition, we have
\begin{align*}
    P_{m,\lambda m}=\frac{\sum_{i=1}^m{P_\lambda^{(i)}}}{m}.
\end{align*}
Since all the subsketches are i.i.d., by the law of large numbers, 
$\lambda\cdot P_{m,\lambda} \to \lambda\cdot \mathbb{E}\paren{P_\lambda^{(1)}}$ 
almost surely as $m\to\infty$.

Let $X$ be the $k$th smallest hash value among $\lambda$ uniformly random numbers in $[0,1]$, which distributes identically with $P_\lambda^{(1)}$. By standard order statistics, $X$ is a Beta random variable $\mathrm{Beta}(k,\lambda -1+k)$ which has mean $\frac{k}{\lambda+1}$. Thus $\kappa_\lambda = \lambda\cdot \E(X)=\frac{k\lambda}{\lambda+1}$. We conclude that
\begin{align*}
    \kappa = \lim_{\lambda\to \infty}\kappa_\lambda=\lim_{\lambda\to \infty} \frac{k\lambda}{\lambda+1}=k.
\end{align*}
\end{proof}

Applying Theorem \ref{thm:kappa} to $(k,m)$-\textsf{MinCount}, 
we see its ARV is $\frac{1}{2km}$,\footnote{For simplicity, we assume the second condition of Theorem 4 holds for all the sketches analyzed in this paper.} matching 
Cohen~\cite{Cohen15} and Ting~\cite{Ting14}.  Technically its
$\MVP$ is unbounded since hash values were real numbers,
but any realistic implementation would store them to 
$\log U$ bits of precision, for a total of $km\log U$ bits.  
Hence we regard its $\MVP$ to be $\frac{1}{2}\cdot \log_2 U$.

\subsubsection{Smoothing Discrete Sketches}\label{sect:smoothing}

Sketches that partition the dartboard in some exponential fashion with base $q$
(like \textsf{LogLog}, $\PCSA$, and \Curtain) have the property that their 
estimates and variance are periodic in $\log_q \lambda$.  
Pettie and Wang~\cite{PettieW21} proposed a simple method to \emph{smooth} these 
sketches and make them truly scale-invariant as $m\to\infty$.

We assume that the dartboard is partitioned into $m$ columns.  The base-$q$ 
\emph{smoothing} operation uses an \emph{offset vector} $\vec{r} = (r_0,\ldots,r_{m-1})$.  
We scale down all the cells in column $i$ by the factor $q^{-r_i}$, 
then add a dummy cell spanning $[q^{-r_i},1)$ which is always occupied.
(Phrased algorithmically, if a dart is destined for column $i$, we filter it
out with probability $1-q^{-r_i}$ and insert it into the sketch with probability $q^{-r_i}$.)
When analyzing variants of (\textsf{Hyper})\textsf{LogLog}
and $\PCSA$, we use the uniform offset vector $(0,1/m,2/m,\ldots,(m-1)/m)$.
The \Curtain{} sketch can be viewed as having a built-in offset vector of 
$(0,1/2,0,1/2,0,1/2,\ldots)$ which effects the ``sawtooth'' cell partition.
To smooth it, we use the offset vector\footnote{In \cite{PettieW21}, 
the smoothing was implemented via \emph{random} offsetting, instead of the \emph{uniform} offsetting. 
In $\Curtain$ we need to use uniform offsetting so that the offset values
of columns are similar to their neighbors.}
\[
(0,\: 1/2,\: 1/m,\: 1/2+1/m,\: 2/m,\: 1/2+2/m,\: \ldots,\: 1/2-1/m,\: 1-1/m).
\]
As $m\to\infty$, $\vec{r}$ becomes uniformly dense in $[0,1]$.

The smoothing technique makes the empirical estimation more scale-invariant
(see~\cite[Figs.~1\& 2]{PettieW21}) but also makes the sketch theoretically 
scale-invariant according to Definition \ref{def:scale-invariant}. 
Thus, in the analysis, we will always assume the sketches are smoothed.
However, in practice it is probably not necessary to do smoothing if $q<3$. 

In the next section, we will prove that \emph{smoothed} 
$\qLL$, $\qPCSA$, and 
$\Curtain$ are all scale-invariant.

\section{Analysis of Dartboard Based Sketches}\label{sect:analysis}

Consider a dartboard cell that covers the vertical interval $[q^{-(t+1)},q^{-t})$.
We define the \emph{height} of the cell to be $t$.
In a smoothed cell partition, no two cells have the same height
and all heights are of the form $t=j/m$, for some integer $j$.
Thus, we may refer to it unambiguously as \emph{cell $t$}.
Note that cell $t$ is an $m^{-1} \times \frac{1}{q^t}\frac{q-1}{q}$ rectangle.

\subsection{Poissonized Dartboard}
Since we care about the asymptotic case where $\lambda\to\infty$, we model the process of ``throwing darts'' by a Poisson point process on the dart board (similar to the ``poissonization'' in the analysis of $\textsf{HyperLogLog}$ \cite{FlajoletFGM07}). Specifically, after throwing $\lambda m$ darts (events) to the dartboard, we assume the number of darts in cell $t$ is a Poisson random variable with mean $\lambda \frac{1}{q^t}\frac{q-1}{q}$ and the number of darts in different cells are independent. For the poissonized dartboard, the range of height of cells naturally 
extend to the whole  set of real numbers, instead of just having cells with positive height.

For any $t\in\mathbb{R}$, let $Y_{t,\lambda}$ be the indicator whether cell $t$ contains at least one dart. 
Note that the probability that a Poisson random variable with mean $\lambda'$ is zero
is $e^{-\lambda'}$. Thus we have,
\begin{align*}
    \Pr(Y_{t,\lambda}=0)=e^{-\frac{\lambda}{q^t}\frac{q-1}{q}}.
\end{align*}

Here, we note some simple identities for integrals that we will use frequently in the analysis.
\begin{lemma}
For any $q>1$, we have
\begin{align*}
    \int \frac{1}{q^t} e^{-\frac{1}{q^t}}dt= \frac{1}{\ln q} e^{-\frac{1}{q^t}}+C.
\end{align*}
Furthermore, let $c_0,c_1$ be any positive numbers, we have
\begin{align*}
    \int_{-\infty}^\infty \frac{c_0}{q^t} e^{-\frac{c_1}{q^t}}dt= \frac{c_0}{c_1} \frac{1}{\ln q}.
\end{align*}
\end{lemma}
\begin{proof}
Use standard calculus.
\end{proof}

\subsection{The Constant $\kappa$}

Let $Z_{t,\lambda}$ be the indicator of whether the cell  $t$ is \emph{free}.
Unlike $Y_{t,\lambda}$, $Z_{t,\lambda}$ depends on which sketching
algorithm we are analyzing.
Since the state changing probability is equal to the sum of the area of free cells, we have
\begin{align}
    P_{m,\lambda m} = \sum_{j=0}^\infty \frac{1}{m}\left(\frac{1}{q^{j/m}}-\frac{1}{q^{{j}/m+1}}\right)Z_{j/m,\lambda}. \label{eq:pmlm}
\end{align}

If $P_{m,\lambda m}$ converges to $\kappa_\lambda/\lambda$ almost surely as $m\to\infty$, then  $\mathbb{E}(P_{m,\lambda m})$ also converges to $\kappa_\lambda/\lambda$ as $m\to\infty$. Thus we have, from (\ref{eq:pmlm}),
\begin{align}
    \kappa_\lambda/\lambda &= \lim_{m\to\infty}\mathbb{E}(P_{m,\lambda m})= \lim_{m\to\infty}\sum_{j=0}^\infty \frac{1}{m}\left(\frac{1}{q^{j/m}}-\frac{1}{q^{{j}/m+1}}\right)\mathbb{E}(Z_{j/m,\lambda})\nonumber\\
    &=\int_{0}^\infty \left(\frac{1}{q^{t}}-\frac{1}{q^{t+1}}\right)\mathbb{E}(Z_{t,\lambda})dt
    \approx \int_{-\infty}^\infty \left(\frac{1}{q^{t}}-\frac{1}{q^{t+1}}\right)\mathbb{E}(Z_{t,\lambda})dt,\label{eq:kappa_lambda}
\end{align}
where we can extend the integration range to negative infinity without affecting the limit of $\kappa_\lambda$ as $\lambda\to\infty$.\footnote{Note that for any $t,\lambda$, we all have $\mathbb{E}(Z_{t,\lambda})\leq \E(1-Y_{t,\lambda})$ (free cell has no dart). Therefore, by extending the integration (\ref{eq:kappa_lambda}) to the whole real line, the increment is bounded by 
$   \int_{-\infty}^0 \left(\frac{1}{q^{t}}-\frac{1}{q^{t+1}}\right)\E(1-Y_{t,\lambda}) dt=\int_{-\infty}^0 \left(\frac{1}{q^{t}}-\frac{1}{q^{t+1}}\right)e^{-\lambda(1/q^t-1/q^{t+1})} dt
    =\frac{1}{\lambda }\int_{-\infty}^{-\log_q(\lambda(q-1)/q)} 1/q^te^{-1/q^t} dt
    =\frac{1}{\lambda}\frac{e^{-\lambda\frac{q-1}{q}}}{\ln q}$ where $\frac{e^{-\lambda\frac{q-1}{q}}}{\ln q}\to 0$
as $\lambda \to\infty$. Thus it will not affect the value of $\lim_{\lambda\to\infty}\kappa_{\lambda}$.}  We conclude that
\begin{align}
    \kappa =\lim_{\lambda\to\infty} \kappa_\lambda =\lim_{\lambda\to\infty}  \lambda\int_{-\infty}^\infty \left(\frac{1}{q^{t}}-\frac{1}{q^{t+1}}\right)\mathbb{E}(Z_{t,\lambda})dt. \label{eq:kappa}
\end{align}
The formula (\ref{eq:kappa}) is novel in the sense that, in order to evaluate $\kappa$, we now only need to understand the probability that $Z_{t,\lambda}$ is 1 for fixed $t$ and $\lambda$.\footnote{Technically, to apply  formula (\ref{eq:kappa}) one needs to first prove that the state changing probability $P_{m,\lambda m}$  converges almost surely to some constant $\kappa_\lambda/\lambda$ for any $\lambda$, which is a mild regularity condition for any reasonable sketch. Thus in this paper we will assume the sketches in the analysis all satisfy this regularity condition and claim that a sketch is scale-invariant if formula (\ref{eq:kappa}) converges.}

\subsection{Analysis of Smoothed $\qPCSA$ and $\qLL$}

The sketches $\qPCSA$ and $\qLL$ are the natural smoothed
base-$q$ generalizations of $\PCSA$~\cite{FlajoletM85} and \textsf{LogLog}~\cite{DurandF03}.

\begin{theorem}\label{thm:PCSA-LL-kappa}
$\qPCSA$ and $\qLL$ are scale-invariant. In particular, we have,
\begin{align*}
    \kappa_{\qPCSA}=\frac{1}{\ln q},\text{ and }\kappa_{\qLL} = \frac{1}{\ln q}\frac{q-1}{q}.
\end{align*}
\end{theorem}
\begin{proof}
For $\qLL$, cell $t$ is free iff both itself and all the cells above it in its column contain no darts. 
Thus we have
\begin{align*}
    \mathbb{E}(Z_{t,\lambda})&=\prod_{i=0}^\infty \Pr(Y_{t+i,\lambda}=0)=\prod_{i=0}^\infty e^{-\frac{\lambda}{q^{t+i}}\frac{q-1}{q}}=e^{-\frac{\lambda}{q^{t}}}.
\end{align*}
Insert it to formula (\ref{eq:kappa}) and we get
\begin{align*}
\kappa_{\qLL} &= \lim_{\lambda\to\infty} \lambda\int_{-\infty}^\infty \left(\frac{1}{q^{t}}-\frac{1}{q^{t+1}}\right)e^{-\frac{\lambda}{q^{t}}}dt= \frac{1}{\ln q}\frac{q-1}{q}.
\end{align*}

For $\qPCSA$, cell $t$ is free iff it has no dart. Thus $Z_{t,\lambda}= 1-Y_{t,\lambda}$ and by formula (\ref{eq:kappa}) we have
\begin{align*}
    \kappa_{\qPCSA}&=\lim_{\lambda\to\infty} \lambda \int_{-\infty}^\infty \left(\frac{1}{q^{t}}-\frac{1}{q^{t+1}}\right)e^{-\frac{\lambda}{q^t}\frac{q-1}{q}} dt= \frac{1}{\ln q}.
\end{align*}
\end{proof}

The \fishmonger~\cite{PettieW21} sketch is based on a smoothed, 
entropy compressed version of base-$e$ $\PCSA$.
The memory footprint of \fishmonger{} approaches its entropy as $m\to\infty$, 
which was calculated to be $mH_0$~\cite[Lemma 4]{PettieW21}. From Theorem \ref{thm:PCSA-LL-kappa}, we know $\kappa_{e\textsf{-PCSA}}=1$.

\begin{cor}
\fishmonger{} has limiting $\MVP$ $H_0/2\approx 1.63$.
\end{cor}
\begin{proof}
By Theorem \ref{thm:kappa}, limiting $\MVP$ equals $mH_0\cdot \frac{1}{2m}=\frac{H_0}{2}$.
\end{proof}

\subsection{Asymptotic Local View}
For any $t$ and $\lambda$, since we want to evaluate $Z_{t,\lambda}$, whose value may depend on its ``neighbors'' on the dartboard, we need to understand the configurations of the cells near cell $t$. Since we consider the case where $m$ goes to infinity, we may ignore the effect of smoothing to the cells in the immediate vicinity of cell $t$.

After taking these asymptotic approximations, we can index the cells near cell $t$ as follows. 
\begin{definition}[neighbors of cell $t$]
Fix a cell $t$. Let $i\in\mathbb{Z}$ and $c\in\mathbb{R}$. The $(i,c)$-neighbor of cell $t$ is a cell whose column index differs by $i$ (negative $i$ means to the left, positive to the right) 
and has height $t+c$, it covers the vertical interval $[q^{-(t+c+1)}, q^{-(t+c)})$. 
In the sawtooth partition, $c$ is an integer when $i$ is even and a half-integer when
$i$ is odd.  (Note that we are locally ignoring the effect of smoothing.)

Once cell $t$ is fixed, define $W(i,c)$ to be the indicator for whether the $(i,c)$-neighbor of cell $t$ 
has at least one dart in it. Thus, for fixed $t,\lambda$, we have
\begin{align*}
    \Pr(W(i,c)=0)=\Pr(Y_{t+c,\lambda}=0)=e^{-\frac{\lambda}{q^{t+c}}\frac{q-1}{q}}.
\end{align*}
\end{definition}

In the asymptotic local view, we lose the property that a cell can 
be uniquely identified by its height, hence the need to refer to nearby 
cells by their position \emph{relative} to cell $t$.

\subsection{Analysis of $\Curtain$}

We first briefly state some properties of curtain. 
For any $a\geq 1$, recall that $\mathscr{O}_a = \{-(a-1/2), -(a-3/2),\ldots, -1/2, 1/2, \ldots, a-3/2,a-1/2\}$. It is easy to see that for any vector 
$v = (g_0,g_1,\ldots,g_{m-1})$, 
$v_{\operatorname{curt}} = (\hat{g}_i)$ can be expressed as
\begin{align*}
    \hat{g}_i &= \max_{j\in[0,m-1]} \{g_j - |i-j|(a-1/2)\}.
\end{align*}
For each $i$, we define the \emph{tension point} $\tau_i$ to be the 
lowest allowable value of $\hat{g}_i$, given the context of its neighboring columns.
\begin{align*}
    \tau_i &= \max_{j\in[0,m-1]\setminus\{i\}} \{g_j - |i-j|(a-1/2)\},
\end{align*}
and thus we have $\hat{g}_i=\max(g_i,\tau_i)$. We see that the column $i$ is \emph{in tension}
iff $g_i\leq \tau_i$, that is, $\hat{g}_i = \tau_i$.

\begin{theorem}\label{thm:curtain-kappa}
$\Curtain$ is scale-invariant with 
\begin{align*}
    \kappa_\Curtain = \frac{1}{ \ln q}\frac{q-1}{q} \frac{1}{\frac{q-1}{q}+\frac{2}{q^{h}(q^{a-1/2}-1)}+\frac{1}{q^{h+1}}}.
\end{align*}
\end{theorem}
\begin{proof}
Fix cell $t$ and $\lambda$. Define $W_1(k)$ to be the height of the highest cell 
containing darts in the column $k$ away from $t$'s column.
I.e., define $\iota = \Indicator{k \mbox{ is odd}}/2$ to be 1/2 if $k$ is odd and zero if $k$ is even,
and
$W_1(k) \bydef \max\{t+i+\iota \mid i\in \mathbb{Z} \mbox{ and } W(k,i+\iota)=1\}$.  

We have for any $i\in\mathbb{Z}$,
\[
\Pr(W_1(k)\leq t+i+\iota)=\prod_{j=1}^\infty \Pr(W(k,i+j+\iota)=0)=e^{-\frac{\lambda}{q^{t+1+i+\iota}}}.
\]
Let $T_1$ be the tension point of the column of cell $t$, which equals $\displaystyle \max_{j\in\mathbb{Z}\setminus\{0\}} \{W_1(j) - |j|(a-1/2)\}$. We have for any $i\in\mathbb{Z}$,
\begin{align*}
    \Pr\left(T_1\leq i+t\right) &= \Pr\left(\max_{j\in\mathbb{Z}\setminus\{0\}}\{W_1(j)-|j|(a-1/2)\}\leq i+t\right) \\
    &=\prod_{j\in\mathbb{Z}\setminus\{0\}}\Pr(W_1(j) - |j|(a-1/2)\leq i+t) \\
    &= \left(\prod_{j=1}^\infty e^{-\lambda \frac{1}{q^{t+i+1+j(a-1/2)}}}\right)^2
    =e^{-\lambda \frac{2}{q^{t+i+1}} \frac{1}{q^{a-1/2}-1}}.
\end{align*}

From the rules of $\Curtain$, we know that a cell is free iff it contains no dart,
it is at most $h-1$ below its column's tension point, 
and at most $h$ below the highest cell in its column containing darts. Thus,
\begin{align*}
    Z_{t,\lambda} = \Indicator{Y_{t,\lambda}=0} \cdot \Indicator{t\geq T_1-(h-1)} \cdot \Indicator{t\geq W_1(0)-h},
\end{align*}
Note that $T_1$ is independent from $Y_{t,\lambda}$ and $W_1(0)$. In addition, $Y_{t,\lambda}$ is also independent from $\Indicator{t\geq W_1(0)-h}$, since the latter only depends on $Y_{t',\lambda}$ with $t'\geq h+t+1$.
Thus, we have
\begin{align*}
    \mathbb{E}( Z_{t,\lambda} )
    &=\Pr(Y_{t,\lambda}=0)\cdot \Pr(T_1\leq t+h-1)\cdot \Pr(W_1(0)\leq t+h)\\
    &=e^{-\frac{\lambda}{q^{t}}\frac{q-1}{q}} e^{-\lambda \frac{2}{q^{t+h}} \frac{1}{q^{a-1/2}-1}}e^{-\frac{\lambda}{q^{t+h+1}}}\\
    &=\exp\left({-\frac{\lambda}{q^{t}}} \left(\frac{q-1}{q}+\frac{2}{q^{h}(q^{a-1/2}-1)}+\frac{1}{q^{h+1}}\right)  \right).
\end{align*}
Thus by formula (\ref{eq:kappa}), we have
\begin{align*}
    \kappa_\Curtain&=\lim_{\lambda\to\infty}\lambda\int_{-\infty}^\infty \left(\frac{1}{q^{t}}-\frac{1}{q^{t+1}}\right)\exp\left({-\frac{\lambda}{q^{t}}} \left(\frac{q-1}{q}+\frac{2}{q^{h}(q^{a-1/2}-1)}+\frac{1}{q^{h+1}}\right)  \right) dt\\
    &=\frac{1}{\ln q}\frac{q-1}{q} \frac{1}{\frac{q-1}{q}+\frac{2}{q^{h}(q^{a-1/2}-1)}+\frac{1}{q^{h+1}}}.
\end{align*}
\end{proof}

\section{Optimality of \Martingale{} \fishmonger}
\label{sect:optimality}
\Martingale{} sketches have several attractive properties, e.g., being strictly unbiased and insensitive to duplicate elements in the data stream.  
In Section~\ref{sect:martingale-optimality} 
we argue that any sketch
that satisfies these natural assumptions
can be systematically transform into 
a \Martingale{} \textsf{X} sketch with equal or 
lesser variance, where 
\textsf{X} is a dartboard sketch.  In other words,
the \Martingale{} \underline{\emph{transform}} is optimal.

In Section~\ref{sect:fishmonger-optimal} we prove that within the class of \emph{linearizable} dartboard sketches, 
\Martingale{} \fishmonger{} is optimal.
The class of linearizable sketches is broad and 
includes state-of-the-art sketches,
which lends strong \emph{circumstantial} evidence 
that the memory-variance product 
of \Martingale{} \fishmonger{} cannot be improved.

\subsection{Optimality of the \Martingale{} Transform}\label{sect:martingale-optimality}

Consider a non-mergeable sketch processing a stream 
$\mathcal{A}=(a_1,a_2,\ldots)$.  Let $S_i$ be its state
after seeing $(a_1,\ldots,a_i)$,
$\lambda_i = |\{a_1,\ldots,a_i\}|$,
and $\hat{\lambda}(S_i)$ be the estimate of cardinality $\lambda_i$
when in state $S_i$.  We make the following natural assumptions.
\begin{description}
\item[Randomness.] The random oracle $h$ is the only source of 
randomness.  In particular, 
$S_i$ is a function of $(h(a_1),h(a_2),\ldots,h(a_i))$.
\item[Duplicates.] If $a_i \in \{a_1,\ldots,a_{i-1}\}$,
$S_i = S_{i-1}$, i.e., duplicates do not trigger state transitions.
\item[Unbiasedness.] Suppose one examines the data structure
at time $i$ and sees $S_i=\mathbf{s}_i$ and then examines
it at time $j$.  Then $\hat{\lambda}(S_j)-\hat{\lambda}(\mathbf{s}_i)$
is an unbiased estimate of $\lambda_j-\lambda_i$.
\end{description}

\begin{definition}
The \emph{state history} at time $i$, 
denoted 
$W_i = (S_0,S_{j_1},S_{j_2},\ldots,S_{j_\ell}=S_i)$,
lists all the \emph{distinct} states
encountered when processing $(a_1,\ldots,a_i)$.
Note that $\ell, \{j_k\}_{1\leq k\leq \ell}$,
and $W_i$ are all random variables.  When we want
to fix a particular state-history $w_i$ we write $W_i=w_i$.
\end{definition}

The \emph{duplicates} and \emph{randomness} assumptions imply that the distribution of $S_i$ and $W_i$ depends only on $\lambda_i$.
Thus, we henceforth assume  
for simplicity that there are 
no duplicates and that $\lambda_i = i$.

\medskip 

Suppose that the algorithm could magically make its 
cardinality estimates based not just on $S_i$,
but the entire state history $W_i$.
Let $\mathcal{H}$ be the (countably infinite) 
set of all histories and 
$\mathbf{H}\in [0,1]^{\mathcal{H}\times\mathcal{H}}$ 
be the stochastic matrix governing
transition between histories.\footnote{I.e., if $w'=(S_0,\ldots,S,S')$
and $w$ is the prefix of $w'$ missing $S'$, then $\mathbf{H}(w,w')$
is the probability that the next (distinct) element would cause
the sketch to transition from $S$ to $S'$.}
Let $e_w \in [0,1]^{\mathcal{H}}$ 
be the probability distribution 
that puts unit probability on history $w$.\footnote{I.e., 
$e_w(w) = 1$ and $e_w(w')=0$ for all $w'\neq w$.}
Let $\hat{\lambda} \in \mathbb{R}^{\mathcal{H}}$
be the vector of cardinality estimates
at each history-state.
From the \emph{Randomness},
\emph{Duplicates}, and 
\emph{Unbiasedness}
assumptions, it follows that if we observe
that $W_i=w_i$, then
\[
e_{w_i}^{\top}\, \mathbf{H}^{j-i}\, \hat{\lambda} = j-i.
\]
is the expectation of $\hat{\lambda}(W_j)-\hat{\lambda}(w_i)$.
Here $e_{w_i}^{\top}\, \mathbf{H}^{j-i}$
is the distribution of the history-state $W_j$ conditioned
on $W_i=w_i$.
In the special case 
that $j=i+1$, we have
\begin{align*}
\E(\hat{\lambda}(W_{i+1}) \mid W_i=w_i)
&=
\hat{\lambda}(w_i) + \E(\hat{\lambda}(W_{i+1})-\hat{\lambda}(w_i) \mid W_{i+1}\neq w_i)\cdot (1-\mathbf{H}(w_i,w_i))\\
\intertext{and due to the \emph{Unbiased} assumption this must be}
&=
\hat{\lambda}(w_i)+1.
\end{align*}
Hence, for any $w_i$,
\begin{align}
\E(\hat{\lambda}(W_{i+1})-\hat{\lambda}(w_i) \mid W_{i+1}\neq w_i)
= (1-\mathbf{H}(w_i,w_i))^{-1}.\label{eqn:opt1}
\end{align}
Phrased differently, the \emph{Unbiased} assumption implies that 
$(X_i)$ is a martingale, where
$X_i = \hat{\lambda}(W_i)-i$.
Define $Z_i = X_i-X_{i-1}$.
Because $(X_i)$ is a martingale the covariances
of the $(Z_i)$ are all zero.  We have
\begin{align}
    \Var(X_i) &= \sum_{j=1}^{i} \Var(Z_j)\nonumber\\
        &= \sum_{j=1}^{i} \left(\E(\Var(Z_j^{\ } \mid W_{j-1})) + \Var(\E(Z_j \mid W_{j-1}))\right)\nonumber
\intertext{Observe that $\E(Z_j \mid W_{j-1})=0$, so this is}
        &= \sum_{j=1}^i \E(\Var(Z_j \mid W_{j-1}))\nonumber\\
        &= \sum_{j=1}^i \sum_{w_{j-1}} \Pr(W_{j-1} = w_{j-1})\cdot \E\left(\left(\hat{\lambda}(W_j) - \hat{\lambda}(w_{j-1})-1\right)^2\right)\label{eqn:opt2}
\intertext{Note that the expression inside the expectation in (\ref{eqn:opt2}) 
is constant when $W_j=w_{j-1}$, which holds with probability
$c_0=\mathbf{H}(w_{j-1},w_{j-1})$.  
Let $c_1,c_2,\ldots$ be other constants that depend only on $w_{j-1}$.  
Continuing, this is equal to}
        &= \sum_{j=1}^i \sum_{w_{j-1}} \left(\E\left(\left(\hat{\lambda}(W_j) - \hat{\lambda}(w_{j-1})-1\right)^2 \;\middle|\; W_j\neq w_{j-1}\right)\cdot c_1 + c_2\right)\label{eqn:opt3}
\intertext{By (\ref{eqn:opt1}), $\E\left(\hat{\lambda}(W_j) - \hat{\lambda}(w_{j-1})-1\right) = (1-\mathbf{H}(w_{j-1},w_{j-1}))^{-1}-1$, which
also depends only on $w_{j-1}$, hence
(\ref{eqn:opt3}) is equal to}
        &= \sum_{j=1}^i \sum_{w_{j-1}} \left(\Var\left(\hat{\lambda}(W_j) - \hat{\lambda}(w_{j-1}) \;\middle|\; W_j\neq w_{j-1}\right)\cdot c_3 + c_4\right).\label{eqn:opt4}
\end{align}

At this point we can ask which
estimate vector $\hat{\lambda}$
minimizes (\ref{eqn:opt4}).
The variances in (\ref{eqn:opt4}) 
are non-negative, and it is 
possible to make them all zero,
subject to (\ref{eqn:opt1}), 
by setting 
\begin{align}
\hat{\lambda}(w_j)
= 
\hat{\lambda}(w_{j-1}) + (1-\mathbf{H}(w_{j-1},w_{j-1}))^{-1}\label{eqn:opt5}
\end{align}
for every $w_j$ such that $\mathbf{H}(w_{j-1},w_j)>0$.
Note that the transitions
in $\mathbf{H}$ that occur with 
non-zero probability, 
excluding self-loops, 
form a directed arborescence (out-tree) rooted at 
the initial history $(S_0)$.  
Thus, all the constraints of the form
(\ref{eqn:opt5}) can be satisfied 
simultaneously.

To recapitulate, as a consequence of the \emph{Randomness}, \emph{Duplicates}, and \emph{Unbiased} assumptions, the
\Martingale{} estimator has minimum variance. Define the \emph{$h$-state} of a sketch state $S$, denoted $\hS$,
to be the set of hash values, that, if encountered, would cause no state transition.
Then we can write $S_i$ as $(\hat{\lambda}_i,\hS_i,\lS_i)$, 
where $\hat{\lambda}_i$ is the \Martingale{} estimate (which depends on the history), 
and $\lS_i$ is any leftover state information not implied by $\hS_i$ and $\hat{\lambda}_i$.
We have shown that $\hat{\lambda}_i$ is
the only information from the history 
useful for making cardinality estimates.
Thus, the data structure is free to 
change $\lS_i$ to any value
consistent with $\hS_i$ at will, 
and therefore $\lS_i$ should
not be stored at all.
In other words, we can simply store the 
state $S_i$ as $(\hat{\lambda}_i,\hS_i)$ and impute any $\lS_i$ which is most advantageous.\footnote{In particular, if $a_{i+1}$ is such that it would cause
$(\hat{\lambda}_i,\hS_i,\lS_i)$ 
to become
$(\hat{\lambda}_{i+1},\hS_{i+1},\lS_{i+1})$
and cause
$(\hat{\lambda}_i,\hS_i,S^{\operatorname{h}\prime}_i)$
to become
$(\hat{\lambda}_{i+1},S^{\operatorname{h}\prime}_{i+1},S^{\operatorname{l}\prime\prime}_{i+1})$, then we are free to choose our next state to be 
$(\hat{\lambda}_{i+1},\hS_{i+1})$
or $(\hat{\lambda}_{i+1},S^{\operatorname{h}\prime}_{i+1})$,
whichever is more advantageous.
As variances improve when $|\hS|$ is small, we would choose the one minimizing 
$|\hS_{i+1}|,|S^{\operatorname{h}\prime}_{i+1}|$.}
  Note that 
since $(\hS_i)$ is a dartboard sketch,\footnote{(occupied cells = hash values that cause no transition)}
$(\hat{\lambda}_i,\hS_i)$ 
is derived by a $\Martingale$ transform
and is not worse than the original sketch $(S_i)$.

\begin{rem}
We should note that under some circumstances it is possible to achieve smaller variance by violating the \emph{duplicates} and \emph{unbiasedness} assumptions. 
For example, suppose the sketch state after seeing $i$ elements were $(\hat{\lambda}_i,S_i,i)$.
If the stream is duplicate-heavy, ``$i$'' carries no useful information, but if nearly all elements are distinct, $i$ is also a good cardinality estimate.  
Since $\lambda_i \leq i$, the cardinality estimate $\min\{\hat{\lambda}_i,i\}$ 
is never worse than $\hat{\lambda}_i$ alone,
but when $\lambda_i\approx i$, it is biased and has a constant factor lower variance.  
\end{rem}

\subsection{Optimality of \Martingale{} \fishmonger{}}\label{sect:fishmonger-optimal}

Given an abstract linearizable sketching scheme \textsf{X}, 
its space is minimized by compressing it to its entropy.
On the other hand, by Theorem~\ref{thm:kappa} 
the variance of \Martingale{} \textsf{X}
is controlled by the normalized expected 
probability of changing state: 
$2\lambda\cdot \E(P_\lambda)$.  Theorem~\ref{thm:core}
lower bounds the ratio of these two quantities for any
sketch that behaves well over a sufficiently large interval
of cardinalities $\lambda \in [e^a,e^b]$.
The proof technique is very similar to~\cite{PettieW21},
as is the take-away message (that \textsf{X}=\fishmonger{} is optimal up to some assumptions).  
However, the two proofs are mathematically distinct 
as \cite{PettieW21} focuses on Fisher information while 
Theorem~\ref{thm:core} focuses on the \emph{probability of state change}.

\begin{theorem}\label{thm:core}
Fix reals $a<b$ with $d=b-a>1$. Let $\Bar{H},\Bar{R}>0$. For any linearizable sketch, let $H(\lambda)$ be the entropy of its state and   $P_\lambda$ be the probability of state change\footnote{The probability of state change $P_\lambda$ is itself a random variable.} at cardinality $\lambda$ satisfies that
\begin{enumerate}
    \item for all $\lambda >0$, $H(\lambda)\leq \Bar{H}$, and
    \item for all $\lambda \in [e^a,e^b]$, $2\lambda \E(P_\lambda) \geq \Bar{R}$,
\end{enumerate}
then 
\begin{align*}
    \frac{\Bar{H}}{\Bar{R}}\geq \frac{H_0}{2} \frac{1-\max(8d^{-1/4},5e^{-d/2})}{1+\frac{(344+4\sqrt{d})}{d}\frac{H_0}{I_0} \left(1-\max(8d^{-1/4},5e^{-d/2})\right)} = \frac{H_0}{2}(1-o_d(1)).
\end{align*}
\end{theorem}
\begin{proof}
By the assumptions of the theorem, we have $\frac{\int_a^b H(e^x) dx }{2\int_a^b e^x\cdot \E(P_{e^x}) dx}\leq \frac{\Bar{H}}{\Bar{R}}$. Thus it suffices to prove that 
\begin{align*}
    \frac{\int_a^b H(e^x) dx }{2\int_a^b e^x\cdot \E(P_{e^x}) dx}\geq \frac{H_0}{2}(1-o_d(1)).
\end{align*}
Now we will write the expressions in terms of the cells.

Let $\mathscr{C}$ be the of cells. By linearizability, we can write cells as $c_0,c_1,\ldots,c_{|\mathscr{C}|-1}$, where $c_i$ has area $p_i$. Let $Z_i$ be the indicator whether $c_i$ is hit by a dart and $Y_i$ be the indicator whether $c_i$ is occupied. Let $F_i=(Y_0,\ldots,Y_i)$. Since it is linearizable, there is some monotone function $\phi:\{0,1\}^*\to \{0,1\}$ such that $Y_i=Z_i\lor \phi(F_{i-1})$. Assume poissonization\footnote{A cell of size $p$ will have probability $e^{-p\lambda}$ to be without a dart at cardinality $\lambda$.}, by Lemma 13 in \cite{PettieW21}, we can write
\begin{align*}
   H(\lambda)=\sum_{i=0}^{|\mathscr{C}|-1}H_B(e^{-p_i\lambda})\Pr(\phi(F_{i-1,\lambda})=0).
   \end{align*}
 Then by the linearity of expectation, we have
   \begin{align*}
\lambda\cdot \E(P_\lambda) = \lambda \sum_{i=0}^{|\mathscr{C}|-1} p_i \Pr(Z_i=0)\Pr(\phi(F_{i-1},\lambda)=0)= \sum_{i=0}^{|\mathscr{C}|-1}p_i\lambda e^{-p_i\lambda}\Pr(\phi(F_{i-1,\lambda})=0).
\end{align*}

For clear presentation, we introduce the following definitions.
\begin{definition}\label{def:bfHR}
Fix a linearizable sketch.
Let $C\subset \mathscr{C}$ be a collection of cells and $W\subset \mathbb{R}$ be an interval
of the reals. Define:
\begin{align*}
    \mathbf{H}(C\to W) &= \int_W \sum_{c_i\in C} \dot{H}(p_ie^x)\Pr(\phi(\mathbf{Y}_{i-1,e^x})=0) dx,\\
    \mathbf{R}(C\to W) &= \int_W \sum_{c_i\in C} \dot{R}(p_ie^x)\Pr(\phi(\mathbf{Y}_{i-1,e^x})=0) dx,
\end{align*}
where 
\begin{align*}
    \dot{H}(t)&= H_B(e^{-t})=\frac{1}{\ln 2}(te^{-t}-(1-e^{-t})\ln (1-e^{-t})),\\
    \dot{R}(t)&=2t e^{-t}.
\end{align*}
Now we can write $\int_a^b H(e^x) dx$ as $\mathbf{H}(\mathscr{C}\to [a,b])$ and $2\int_a^b e^x\cdot \E(P_{e^x}) dx$ as $\mathbf{R}(\mathscr{C}\to [a,b])$. 
\end{definition}
Note that $\int_{-\infty}^\infty \dot{R}(e^{x}) dx = 2$ and it is proved in \cite{PettieW21} that $\int_{-\infty}^\infty \dot{H}(e^{x}) dx = H_0$.\footnote{$H_0=(\ln 2)^{-1}+\sum_{k=1}^\infty k^{-1}\log_2(1+1/k)$.} Thus we want to prove  $\frac{ \mathbf{H}(\mathscr{C}\to [a,b])}{\mathbf{R}(\mathscr{C}\to [a,b])} \geq \frac{\int_{-\infty}^\infty \dot{H}(e^{x}) dx}{\int_{-\infty}^\infty \dot{R}(e^{x}) dx}(1-o_d(1))$. A similar statement is proved in Theorem 5 in \cite{PettieW21} where it is showed that $\frac{ \mathbf{H}(\mathscr{C}\to [a,b])}{\mathbf{I}(\mathscr{C}\to [a,b])} \geq \frac{\int_{-\infty}^\infty \dot{H}(e^{x}) dx}{\int_{-\infty}^\infty \dot{I}(e^{x}) dx}(1-o_d(1))$ and $\mathbf{I}$ and $\dot{I}$ are defined as follows.
\begin{align*}
    \mathbf{I}(C\to W) = \int_W \sum_{c_i\in C} \dot{I}(p_ie^x)\Pr(\phi(\mathbf{Y}_{i-1,e^x})=0) dx,
\end{align*}where $\dot{I}(t)=\frac{t^2}{e^t-1}$. Note that the only difference between this theorem and Theorem 5 in \cite{PettieW21} is between $\dot{R}$ and $\dot{I}$. However, one can verify that $\dot{R}(t)$ satisfies all the properties (see the lemma below) it is used for $\dot{I}(t)$ in the proof of Theorem 5 in \cite{PettieW21}. Thus the similar lower bound is obtained here. 
\end{proof}

\begin{lemma}
The following statements are true. 
\begin{enumerate}
    \item $\frac{\dot{H}(t)}{\dot{R}(t)}$ is decreasing in $t$ on $(0,\infty)$. This corresponds to Lemma 12 in \cite{PettieW21}.
    \item $\dot{R}(t)\leq 4e^{-t/2}$ for all $t>0$. This corresponds to Lemma 20 in \cite{PettieW21}. 
\end{enumerate}
\end{lemma}
\begin{proof}
\begin{enumerate}
    \item $\frac{\dot{H}(t)}{\dot{R}(t)} = \frac{1}{2\ln 2} (1-\frac{(1-e^{-t})\ln (1-e^{-t})}{t e^{-t}})$. Note that $\frac{(1-e^{-t})\ln (1-e^{-t})}{-t e^{-t}}=\frac{1-e^{-t}}{t}\cdot \frac{-\ln(1-e^{-t})}{e^{-t}}$. Let $g(t)=\frac{1-e^{-t}}{t}$ and $h(x)=-\frac{\ln(1-x)}{x}$ where $t>0$ and $x\in (0,1)$. It suffices to prove that $g(t)$ is decreasing and $g(x)$ is increasing. Indeed, $g'(t)=\frac{e^{-t}t-1+e^{-t}}{t^2}< \frac{e^{-t}e^t-1}{t^2}=0$ since $t+1< e^t$ for  $t>0$; $h'(x)=\frac{\frac{x}{1-x}+\ln(1-x)}{x^2}>0$ for $x\in (0,1)$.
    \item It suffices to prove $t/2\leq e^{t/2}$, which is true.
\end{enumerate}

\end{proof}

\begin{cor}
The MVP of any linearizable and scale-invariant sketch is at least $\frac{H_0}{2}$.
\end{cor}
\begin{proof}
Let $S$ be a scale-invariant combined sketch with constant $\kappa$. First recall Definition \ref{def:scale-invariant} that for any $\lambda$, we have $\lambda \cdot P_{m,\lambda m}$ converges to $\kappa_\lambda$ almost surely as $m\to \infty$ where $P_{m,\lambda m}$ is the updating probability after $\lambda m$ insertions to a combined sketch consisting of $m$ base-sketches\footnote{The combined sketch is assumed to be smoothed.}. By the dominated convergence theorem, we have $\lambda\cdot\E(P_{m,\lambda m})$ converges to $\kappa_\lambda$. Furthermore, recall that $\kappa$ is the limit of $\kappa_\lambda$ as $\lambda \to \infty$.  Therefore, for any $\epsilon_1>0$, there exist sufficiently large $m_1$ and $\lambda_1$, such that for any $\lambda > \lambda_1$, $\lambda\cdot \E(P_{m_1,\lambda m_1})> \kappa - \epsilon_1$.

Now by Theorem \ref{thm:kappa} and the definition of ARV factor (Definition \ref{def:ARV}), for any $\epsilon_2>0$, there exist sufficiently large $m_2>m_1$ and $\lambda_2> \lambda_1$, such that for any $\lambda > \lambda_2$, 
$m_2\cdot \frac{\Var(E_{m_2,\lambda m_2})}{(\lambda m_2)^2}>\frac{1}{2\kappa}-\epsilon_2> \frac{1}{2\lambda \E(P_{m_2,\lambda m_2})+2\epsilon_1} -\epsilon_2$ where $E_{m_2,\lambda m_2}$ is the \Martingale{} estimator.

Then, fixing $m_2$, view the combined sketch as \emph{a single sketch}\footnote{Thus $E_{m_2,\lambda m_2}$ should be written as $E_{\lambda m_2}$ and $P_{m_2,\lambda m_2}$ should be $P_{\lambda m_2}$.} and we have that for any $\lambda > \lambda_2 m_2$, $\frac{\Var(E_{\lambda})}{\lambda^2}>\frac{1}{2 \lambda \E(P_\lambda) +2\epsilon_1 m_2} - \frac{\epsilon_2}{m_2}.$ Suppose the sketch uses $\Bar{H}$ bits of memory and for sufficiently large $\lambda$, the relative variance is bounded by $\delta$. Thus we have 
\begin{align*}
    \delta \geq \frac{\Var(E_\lambda)}{\lambda^2} > \frac{1}{2 \lambda \E(P_\lambda) +2\epsilon_1 m_2} - \frac{\epsilon_2}{m_2},
\end{align*}
which says
\begin{align*}
    2\lambda \E(P_\lambda)\geq \frac{1}{\delta+\frac{\epsilon_2}{m_2}}-2\epsilon_1 m_2.
\end{align*}
Invoking Theorem \ref{thm:core} where $a$ and $b$ can be chosen arbitrarily far away, as long as $b>a>\log \lambda_2$, we have 
\begin{align*}
    \frac{\Bar{H}}{\frac{1}{\delta+\frac{\epsilon_2}{m_2}}-2\epsilon_1 m_2}\geq \frac{H_0}{2}.
\end{align*}
Finally note that $\epsilon_1,\epsilon_2$ can be made arbitrarily small and $m_2$ can be made arbitrarily large. We conclude that the MVP $\Bar{H}\cdot \delta \geq \frac{H_0}{2}$.
\end{proof}

\section{Experimental Validation}\label{sect:experiments}



Throughout the paper we have maintained a possibly unhealthy devotion to asymptotic analysis, taking $m\to \infty$ whenever it was convenient. In practice $m$ will be a constant, and possibly a smallish constant. How do the sketches perform in the pre-asymptotic region?

In turns out that the theoretical analysis predicts the performance of \Martingale{} sketches pretty well, even whem $m$ is small. In the experiment of Figure~\ref{fig:SuperCompresionDistribution}, we fixed the sketch size at a tiny $128$ bits. Therefore \textsf{HyperLogLog} uses $m_1 = \floor{128/6} = 21$ counters. The \Martingale{} \textsf{LogLog} and \Martingale{} \Curtain{} sketches encode the martingale estimator with a floating point \emph{approximation} of $\hat{\lambda}$ in 14 bits, with a 6-bit exponent and 8-bit mantissa. Thus, \Martingale{} \textsf{LogLog} uses $m_2 = (128-14)/6 = 19$ counters,  and \Martingale{} \Curtain{} uses $m_3 = 37$.\footnote{It uses the optimal parameterization $(q,a,h)=(2.91,2,1)$
of Theorem~\ref{thm:martingale-curtain}.}

For larger sketch sizes, the distribution of $\hat{\lambda}/\lambda$ 
is more symmetric, 
and closer to the predicted performance.  
Figure~\ref{fig:1200bits} gives the empirical
distribution of $\hat{\lambda}/\lambda$ 
over 100,000 runs when $\lambda=10^6$ and the sketch size
is fixed at 1,200 bits. 
Here \Martingale \Curtain{} uses $m=400$, 
and both \Martingale{} \textsf{LogLog} 
and \textsf{HyperLogLog} use $m=200$.

The experimental and predicted relative variances and standard errors are given in Table~\ref{tab:variances}.

\begin{figure}[h]
    \centering
    \begin{minipage}[t]{0.49\textwidth}
    \includegraphics[width=\textwidth]{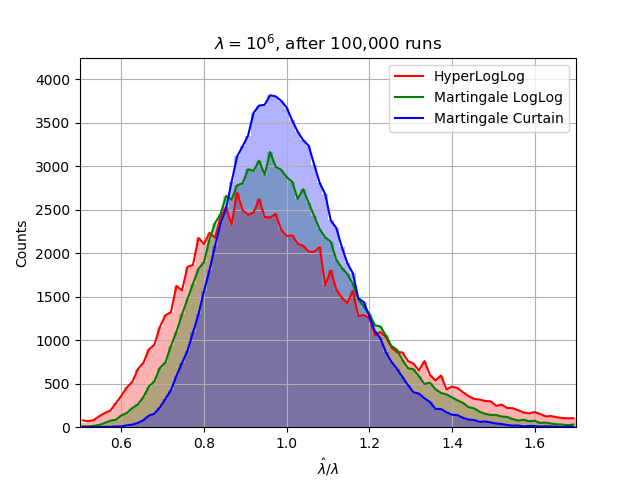}
    \caption{The sketch size is fixed at $128$ bits.}
    \label{fig:SuperCompresionDistribution}
    \end{minipage}
    \begin{minipage}[t]{0.49\textwidth}
    \includegraphics[width=\textwidth]{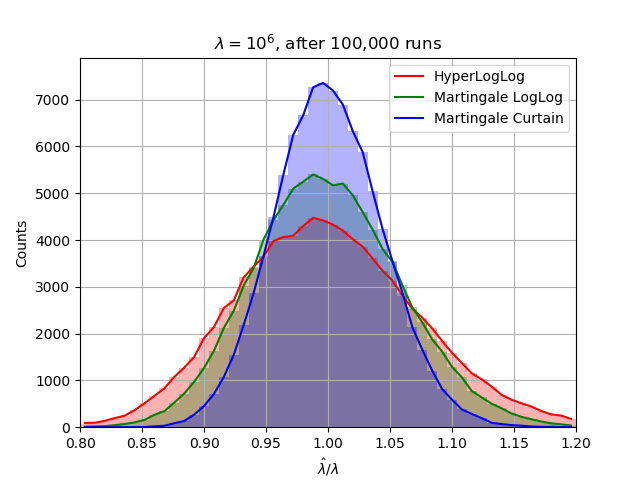}
    \caption{The sketch size is fixed at 1200 bits.}
    \label{fig:1200bits}
    \end{minipage}
\end{figure}

\begin{table}[h]
    \centering
    \scalebox{.95}{
    \begin{tabular}{|l|l|l|l|l|l|l|l|l|}
\hline\hline
\multirow{3}{*}{\textsc{Sketch}} &
\multicolumn{4}{l|}{Using $128$ bits} &
\multicolumn{4}{l|}{Using $1200$ bits}
\\\cline{2-9}
&\multicolumn{2}{l|}{Experiment}&
\multicolumn{2}{l|}{Prediction}&
\multicolumn{2}{l|}{Experiment}&
\multicolumn{2}{l|}{Prediction}\\
\cline{2-9}
&\multicolumn{1}{l|}{Var}&
\multicolumn{1}{l|}{StdErr}&
\multicolumn{1}{l|}{Var}&
\multicolumn{1}{l|}{StdErr}&
\multicolumn{1}{l|}{Var}&
\multicolumn{1}{l|}{StdErr}&
\multicolumn{1}{l|}{Var}&
\multicolumn{1}{l|}{StdErr}
\\\hline
\textsf{HyperLogLog}  \hfill &
$0.0573$ & $23.94\%$ & $0.0549$ & $23.44\%$   & $0.00541$ & $7.36\%$   & $0.00539$ & $7.35\%$  \\
\Martingale{} \textsf{LogLog} \hfill &
$0.0348$ & $18.65\%$   & $0.0365$ &  $19.10\%$  & $0.00350$ &  $5.91\%$  & $0.00347$ & $5.89\%$\\
\Martingale{} \Curtain{} \hfill &
$0.0211$ &  $14.54\%$  & $0.0208$ & $14.43\%$   & $0.00189$ &  $4.35\%$  & $0.00193$ & $4.39\%$\\
\hline\hline
    \end{tabular}
}
    \caption{The relative variance is $\frac{1}{\lambda^2}\Var(\hat{\lambda} \mid \lambda)$ and standard error is $\frac{1}{\lambda}\sqrt{\Var(\hat{\lambda} \mid \lambda)}$.
    The predictions for \Martingale{} \textsf{LogLog} and \Martingale{} \Curtain{}
    use Theorems~\ref{thm:kappa}, \ref{thm:PCSA-LL-kappa}, and~\ref{thm:curtain-kappa}.  The predictions for \textsf{HyperLogLog}
    are from Flajolet et al.~\cite[p. 139]{FlajoletFGM07}.}\label{tab:variances}
\end{table}


\section{Conclusion}\label{sect:conclusion}

The \Martingale{}~transform is attractive due to its simplicity and low variance, but it results in \emph{non}-mergeable sketches. 
We proved that under natural assumptions\footnote{(insensitivity to duplicates, and unbiasedness)},
it generates optimal estimators 
automatically, allowing one to design
structurally more complicated sketches, 
without having to worry about designing or analyzing \emph{ad hoc} estimators.
We proposed the \Curtain{} sketch, 
in which each subsketch only needs a constant number of bits of memory, for \emph{arbitrarily large} cardinality $U$.\footnote{Note that an $O(\log\log U)$-bit offset register is needed for the whole sketch.}

The analytic framework of 
Theorems~\ref{thm:me} and \ref{thm:kappa} simplifies Cohen~\cite{Cohen15} and Ting~\cite{Ting14},
and gives a user-friendly formula for the asymptotic relative variance (ARV)
of the \Martingale{} estimator, as a function of the sketch's constant $\kappa$. 
We applied this framework to \Martingale{} \Curtain{} as well as the \Martingale{} version of the classic sketches ($\mathsf{MinCount}$, $\mathsf{HLL}$ and $\mathsf{PCSA}$).

Assuming perfect compression, one gets the \emph{memory-variance product} ($\MVP$) of an sketch by multiplying its entropy and ARV.
It is proved that for linearizable sketches, \fishmonger{} is optimal for mergeable sketches~\cite{PettieW21}
(limiting $\MVP = H_0/I_2 \approx 1.98$). 
In this paper we proved that in the sequential (non-mergeable)
setting, if we restrict our attention 
to linearizable sketches, 
that \Martingale{} \fishmonger{} is optimal, 
with limiting $\MVP = H_0/2 \approx 1.63$ (Section \ref{sect:fishmonger-optimal}). 
We conjecture that these two lower bounds hold for general, possibly \emph{non}-linearizable sketches.

\bibliographystyle{abbrv}
\bibliography{ref}

\end{document}